\documentclass[10pt,journal]{IEEEtran}
\IEEEoverridecommandlockouts
\usepackage[utf8]{inputenc} 
\usepackage[T1]{fontenc}
\usepackage{url}
\usepackage{ifthen}
\usepackage{cite}
\usepackage[cmex10]{amsmath} 

\usepackage{textcase}
\usepackage{blindtext}
\usepackage{floatrow}
\usepackage{soul}
\usepackage{changes}
\usepackage{gensymb}
\usepackage{cite}
\usepackage{amsmath,amssymb,amsfonts}
\usepackage{textcomp}
\usepackage{graphicx}
\usepackage{amssymb}
\usepackage{amsfonts}
\usepackage{amsmath}

\usepackage{amsthm}
\theoremstyle{definition}

\usepackage{epsfig}
\usepackage{color}
\usepackage{fancybox}
\usepackage{textcomp}
\usepackage{multirow}
\usepackage{setspa ce}
\usepackage{psfrag}
\usepackage{booktabs}
\usepackage{float}
\usepackage{stfloats}
\usepackage{xspace}
\usepackage[caption = false]{subfig}
\usepackage{mathtools, nccmath, bigints}
\usepackage{array}
\usepackage{mathrsfs}

\newcommand{\vect}[1]{\mathbf{#1}}

\def\Htran{\mbox{\tiny $\mathrm{H}$}}

\newtheorem{Theorem}{Theorem}

\newtheorem{Observation}{Observation}

\newtheorem{Corollary}{Corollary}

\IEEEaftertitletext{\vspace{-0.9\baselineskip}}

\begin{document}

\title{Finite Beam Depth Analysis for Large Arrays}

\author{Alva~Kosasih,~Emil~Bj{\"o}rnson,~\IEEEmembership{Fellow,~IEEE}  
\thanks{Parts of this paper have been  {accepted} for a publication at the 2023 IEEE Global Communications Conference (Globecom) \cite{2023_Alva_Globcomm}. A. Kosasih and E. Bj{\"o}rnson are with the Division of Communication Systems,  KTH Royal Institute of Technology, Stockholm, Sweden. E-mail: \{kosasih,emilbjo\}@kth.se. This paper was supported by the Grant 2019-05068 from the Swedish Research Council.}}

\maketitle

\begin{abstract}
Most wireless communication systems operate in the far-field region of antennas and antenna arrays, where waves are planar and beams have infinite depth. When antenna arrays become electrically large, it is possible that the receiver is in the radiative near-field of the transmitter, and vice versa. Recent works have shown that near-field beamforming exhibits a finite depth, which enables a new depth-based spatial multiplexing paradigm. In this paper, we explore how the shape and size of an array determine the near-field beam behaviors. In particular, we investigate the $3$ dB beam depth (BD), defined as the range of distances where the gain is greater than half of the peak gain. We derive analytical gain and BD expressions and prove how they depend on the aperture area and length. For non-broadside transmissions, we find that the BD increases as the transmitter approaches the end-fire direction of the array.  Furthermore, it is sufficient to characterize the BD for a broadside transmitter, as the beam pattern with a non-broadside transmitter can be approximated by that of a smaller/projected array with a broadside transmitter. Our analysis demonstrates that the BD can be ordered from smallest to largest as ULA, circular, and square arrays.
\end{abstract}


\section{Introduction}

Thanks to the successful implementation of massive multiple-input multiple-output (M-MIMO) in $5$G systems in sub-$6$ GHz and mm-wave bands \cite{2013_Rappaport_Access}, it is likely that even larger arrays (e.g., extremely large aperture arrays \cite{2018_Amiri_Globcomm,2020_Wang_TWC} and holographic MIMO \cite{2019_Pizzo_Arxiv}) and wider spectrum at higher frequencies \cite{2019_Rappaport_Access,2020_Sanguinetti_JSAC} will be employed in the next generation of wireless systems \cite{2019_Björnson_DSP,2020_JZhang_JSAC}. A higher carrier frequency implies a smaller wavelength and, therefore, a smaller antenna size. As the array's aperture grows and the wavelength shrinks, the traditional far-field distance boundary, known as the Fraunhofer distance, will be very large since it is proportional to the squared aperture length and inversely proportional to the wavelength. 
Hence, future wireless systems are unlikely to operate in the far field and the signal processing must be redesigned to not rely on the far-field assumption. 
For instance, far-field configured arrays feature significant array gain degradation when used in the near-field \cite{2021_Tang_TWC}.

Three regions have been defined to classify the wavefront behavior with respect to the distance between the transmitter and receiver  \cite{2022_Ramezani_BookChapter}: the reactive near-field, the radiative near-field, and the far-field \cite{1962_Sherman_TAP}. The latter two are of interest for long-range communications, as in mobile networks. 
At the beginning of the radiative near-field, there are discernible and unavoidable amplitude variations 
over the wavefront, but the main property is the spherical phase variations because these can be managed through refined signal processing. In the far-field, both the amplitude and phase variations are insignificant. 
The studies \cite{2022_Lu_TWC,2020_Torres_Asilomar,2023_Bacci_WCL} demonstrate the detrimental impact of using the far-field approximation in scenarios where near-field model is needed. They have also provided the performance analysis employing metrics such as SNR and spectral efficiency, within the framework of a near-field model. It is known that the near-field phenomena can lead to higher performance when exploited for interference suppression, but the specific qualitative reasons for this have not been analyzed in detail, except for being attributed to having spherical wavefronts. 
A deeper understanding of the beam-focusing ability is important for strategic beam arrangement to limit interference, which is the emphasis of our study, codebook design, selection of array geometry, etc.

\subsection{Related Works}

Previously negligible propagation phenomena become dominant in the near-field.
In conventional far-field beamforming, the signal in line-of-sight scenarios is focused on a point that is infinitely far away, resulting in a directive beam with a limited angular width but continuing to infinity. In the near-field, the array acts like a  lens, concentrating the signal onto a particular location rather than directing it to a specific angle, which is the behavior in the far-field.
This is known as beamforming in the distance domain or finite-depth beamforming  \cite{2019_Björnson_DSP, 2021_Yang_TWC,2022_Jiang_TWC,2023_Wu_Arxiv,2022_Myers_TWC,2023_Deshpande_WCL,2021_Björnson_Asilomar}, which is realized by utilizing a matched filter that takes into consideration the phase shift of each individual element present in the array.  In essence, near-field beamforming enables one to configure the beam not only in the angular domain (beam width) but also in the distance domain (beam depth). 
Finite-depth beamforming can potentially pave the way for a novel multiplexing approach, wherein multiple users situated in identical angular directions relative to the ELAA, but at varying distances, can be concurrently accommodated.    
Multiplexing in the joint distance and angle domains introduces a revolutionary approach to efficiently serve crowds of users, a challenge that traditional far-field beamforming struggles to effectively address.

The concept of finite-depth near-field beamforming was first introduced in \cite{2019_Björnson_DSP}. The fundamental properties can be studied using continuous matched filtering, representing the case when the array consists of many small elements so that summations over the antennas can be expressed using integrals \cite{Hu2018a,Dardari2020a,2020_Björnson_JCommSoc}. Since then, researchers have conducted numerous studies utilizing the limited-depth capability of near-field beamforming. In \cite{2021_Yang_TWC}, the focusing property was utilized to derive the closed-form array response, which is then used to develop an effective method for location and channel parameters estimation. Similarly, \cite{2022_Jiang_TWC} proposed the focal scanning method for sensing the location of the receiver, based on the analytic expression of the reflection coefficient of a reflective intelligent surface. 
\cite{2023_Wu_Arxiv} proposed location division multiple access (LDMA) to enhance spectrum efficiency compared to classical spatial division multiple access (SDMA). The LDMA exploited extra spatial resources in the distance domain to serve users at different locations in the near field. 
\cite{2022_Myers_TWC,2023_Deshpande_WCL} further analyzed the near-field beam for massive wideband phased arrays. The focus of their work was to develop a low-complexity technique for constructing beams, which is particularly suitable for massive wideband phased arrays. 

Despite the growing interest in near-field beamforming, there have been no investigations reported on the near-field beam radiation pattern in the distance domain, referred to as the beam depth (BD), except for the recent study  \cite{2021_Björnson_Asilomar}.  That paper characterizes the BD in the broadside direction of a square array, including the distance interval where the array gain is within $3$ dB from the peak value. In contrast to beam width analysis which serves as a guideline for multiplexing users in the angular domain, beam depth analysis allows us to characterize beam patterns in the distance domain, thereby expanding the potential for multiplexing not only in the angular domain but also in the distance domain. The finite-depth beam behavior in the near-field is a unique new feature that enhances spatial resolution. It was previously unknown how the beam depth is affected by the array geometry and transmission direction. By bridging this gap, the antenna array geometry can be optimized to make the most out of the finite BD property of near-field beamforming.

\subsection{Contributions}

In this paper, our objective is to extend the understanding of near-field beamforming by considering a rectangular array model with a tunable width-to-height proportion. In particular, we investigate the array gain and BD. 
We analytically derive the array gain and BD using the Fresnel approximation, which is demonstrated to be an accurate and reliable method in our analysis. Using these derived expressions, we explore how various geometrical factors, such as the width-to-height proportion, area, and aperture length, affect the array gain and BD of rectangular arrays.  We further consider rectangular arrays with a non-broadside transmitter and analyze how the array gain and BD are affected by the  azimuth angle. Moreover, we extend the analysis by considering near-field beamforming with a circular array, and showcase the potential advantages and limitations compared with rectangular arrays. Finally, to better understand the beam behaviors in the distance domain, we analyze the nulls and sidelobes appearing as a function of the propagation distance. 
There are four main contributions of this paper:
\begin{itemize}
    \item We analyze the BD for broadside beamforming with respect to the rectangular array shape and size.  {We derive the normalized array gain, compute the BD analytically, and obtain the upper finite BD limit.} The analysis allows us to characterize the finite-depth beam behavior in the near-field and how the array geometry can be tuned to achieve preferred properties. This analysis can be used to define beam orthogonality in the spatial domain.
    \item  We extend the analysis to consider non-broadside beamforming, in which case direct Fresnel approximations are highly inaccurate. We refine the approximation to obtain high accuracy and then derive the analytical array gain with respect to the azimuth angle. Furthermore, we analyze the BD with respect to the azimuth angle.
    \item We analyze the BD for a circular array. The analysis is performed by deriving the normalized array gain, computing the BD analytically, and obtaining the upper finite BD limit. Furthermore, we characterize the beam behaviors outside the mainlobe in the distance domain by analyzing the nulls and sidelobes as a function of the propagation distance. The results provide insights into how such sidelobes may interfere with the other users' signals.
    \item  {We show that the exact gain of a non-broadside transmitter array can be approximated by the gain of a smaller array with a broadside transmitter by projecting it onto a plane perpendicular to the axis of the non-broadside transmitter.}
\end{itemize}
The paper is organized as follows. In Section~\ref{Sect_Prelim}, we provide preliminaries and definitions for an individual antenna and antenna array. In Section~\ref{Sect_Beamdepth_Rect_Broadside}, we derive the gain and BD of rectangular arrays with a broadside transmitter and analyze its properties. In Section~\ref{Sect_BD_Rect_NonBroadside}, we analyze the gain and BD of rectangular arrays with a non-broadside transmitter.  In Section~\ref{Sect_Beamdepth_Circ}, we derive the gain and BD of a circular array.  This paper is concluded in Section~\ref{Sect_conclusion}.

 \vspace*{-0.05cm}

\section{Preliminaries}
\label{Sect_Prelim}

In this section, we discuss the behavior of electromagnetic radiation from a passive antenna and an antenna array. 
Three regions classically define the electromagnetic radiation patterns with respect to the propagation distance: the reactive near-field, radiative near-field, and far-field. 

\vspace*{-0.2cm}
\subsection{Gain of a Passive Antenna}

The reactive near-field of a large antenna with an aperture length of $D$ spans from $z=0$ to $z\approx 0.62 \sqrt{D^3/\lambda}$ \cite{2021_Björnson_Asilomar}, where $z$ is the distance from the antenna and $\lambda$ is the wavelength. The radiative near-field, also known as the Fresnel region, starts roughly where the reactive near-field ends, i.e.,  from  $z=1.2D$  to $z= d_F$ \cite{2021_Björnson_Asilomar}, where $d_F =  \frac{2D^2}{\lambda}$ is the Fraunhofer distance \cite{1962_Sherman_TAP}, obtained by using a Taylor approximation with an assumption that a phase difference of $\frac{\pi}{8}$ over the aperture is negligible when analyzing the gain. 

These definitions are based on classical approximation, which might not be well-tuned for future systems. It is therefore essential to begin our new analysis from electric field expressions.
We consider a free-space propagation scenario, where there exists an isotropic transmitter located at $\left(x_t,y_t,z \right)$ and a planar receiver centered at the origin covering an area denoted as $\mathcal{A}$ in the $xy$-plane. When the transmitter emits a signal with polarization in the $y$-dimension, the electric field at the point $(x,y,0)$ of the receiver aperture becomes \cite[App. A]{2020_Björnson_JCommSoc}
\begin{equation}\label{eq_II_ElectField}
 E (x,y) =  \frac{E_0}{\sqrt{4 \pi}} \frac{\sqrt{z (\left( x-x_t \right)^2+z^2)}} 
 { r ^{5/4}}   e^{-j\frac{2\pi}{\lambda} \sqrt{r}},
\end{equation}
where $r =  ( ( x-x_t)^2 + (y-y_t)^2 + z^2 )$ is the squared Euclidean distance between the transmitter and the considered point, and $E_0$ is the electric intensity of the transmitted signal in the unit of Volt. The electric field expression in \eqref{eq_II_ElectField} is valid in the Fresnel region since it takes into account the reduction in effective area due to the incident angle, polarization mismatch, and free-space geometric path loss. The complex-valued channel response  to the receive antenna can then be calculated as in \cite[Eq. (64)]{2020_Björnson_JCommSoc}, $h = \frac{1}{\sqrt{A}E_0} \int_{\mathcal{A}} E(x,y)dx dy$.
Hence, the received power is computed as $E_0^2 |h|^2/\zeta$, where $\zeta$ is the impedance of free space. The receive antenna gain compared to an isotropic antenna is computed as \cite[Eq. (6)]{1960_Kay_TAP} $G = \frac{|\int_{\mathcal{A}} E(x,y) dx dy |^2 }{\frac{\lambda^2}{4\pi} \int_{\mathcal{A}} |E(x,y)|^2 dx dy  }$, where $\frac{\lambda^2}{4\pi}$ is the area of the isotropic antenna. 
The antenna gain is normally analyzed in the far-field where the incident wave is plane so that $G_{\rm plane}= \frac{4\pi}{\lambda^2} A$, where $A$ is the physical area of the antenna element. A different value is obtained in the Fresnel region, where there can be phase and amplitude variations over the aperture.
In this paper, we consider the normalized antenna gain defined in \cite{2021_Björnson_Asilomar} as
\begin{equation}\label{eq_II_NormalizedAntennaGain}
    G_{\rm antenna} = \frac{G}{G_{\rm plane}} = \frac{\left| \int_{\mathcal{A}} E(x,y) dx dy \right|^2}{A \int_{\mathcal{A}} \left|E(x,y)\right|^2 dx dy },
\end{equation}
which compares the antenna gain in the Fresnel region to the one obtained in the far field.

\vspace*{-0.2cm}

\subsection{Gain of an Antenna Array}

In this paper, we consider a uniform planar array (UPA) with $N$ antenna elements, where $N$ can be a large number. The antenna elements are uniformly distributed across the array and we consider the same number of antenna elements in the $x$ and $y$ dimensions. However, the elements can be rectangular, which will result in a rectangular aperture shape in the $xy$-plane. Each antenna element is indexed by $n \in \{1, \dots, \sqrt{N} \}$ and  $m \in \{1, \dots, \sqrt{N} \}$ which are the location indexes on the $x$ and $y$ axes, respectively. The Fraunhofer array distance depends on the antenna/array sizes and becomes $d_{FA} = Nd_F$.

We consider the same isotropic transmitter as before but assume that the receiver is equipped with the UPA. The complex-valued channel response to a receive antenna element can then be written 
\begin{equation}\label{eq_II_ArrayChResponse}
h_{n,m} = \frac{1}{\sqrt{A}E_0} \int_{\mathcal{A}} E(x,y)dx dy,
\end{equation}
where $A$ denotes the physical area of each antenna. Similarly to \eqref{eq_II_NormalizedAntennaGain}, we  define the normalized array gain 
\begin{equation}\label{eq_II_NormalizedArrayGain}
    G_{\rm array} = \frac{ \sum_{n=1}^{\sqrt{N}} \sum_{m=1}^{\sqrt{N}} \left| \int_{\mathcal{A}} E(x,y) dx dy \right|^2}{N A \int_{\mathcal{A}} \left|E(x,y)\right|^2 dx dy },
\end{equation}
which is the combined antenna and array gains achieved in the Fresnel region compared to what is achievable in the far-field. By combining the antenna elements in the described way, we obtain a continuous surface and can study the radiated field and the resulting array gain for any continuous charge distribution over it. Practical arrays typically consist of discrete elements and it is not straightforward to use them to generate a desired charge distribution since mutual coupling makes the actual distribution different from the antenna input.
One solution approach is to make a meticulous metal patch design, allowing control over partially diffracted waves from the array-antenna decoupling surface to counteract undesired coupled waves, thereby maintaining an acceptable antenna pattern distortion \cite{2022_Wang_Globecom}. Another approach is to incorporate a decoupling matrix in the beamforming design, which can compensate for the mutual coupling \cite{2018_Chen_Access}. Our paper will not dive into these implementation aspects but focus on what is fundamentally achievable.

\section{Gain and Beam Depth of Rectangular Arrays with a Broadside Transmitter}\label{Sect_Beamdepth_Rect_Broadside}

\begin{figure}
\centering
{\includegraphics[width=0.8\textwidth]{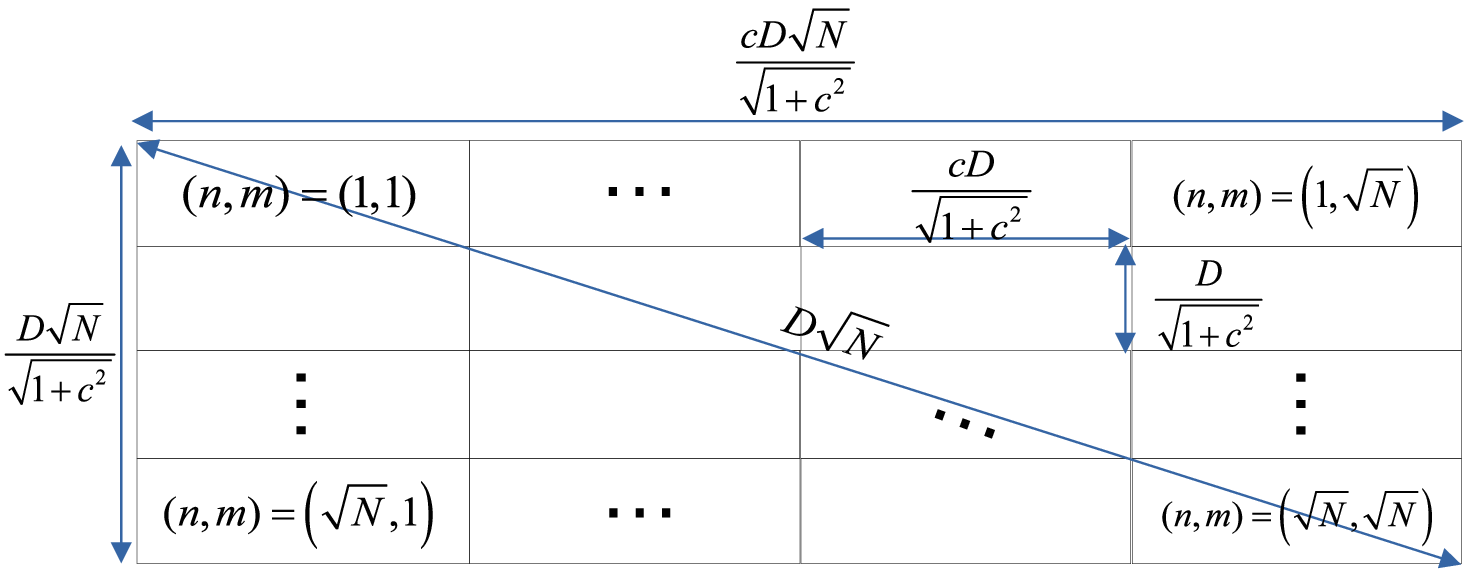}}
\caption{We consider a generalized rectangular array design, where the width of each antenna element is $ {\eta}$ times its height. The number of elements is constant irrespective of the sizes or shapes of the elements. }
\label{Fig_rect_design}
\end{figure}

In this section, we consider a UPA with a rectangular shape. As illustrated in Fig.~\ref{Fig_rect_design}, the array is composed of antenna elements, each with a height of $\ell = D/\sqrt{(1+ {\eta}^2)}$ and width of $w= {\eta}  \ell$ for some $ {\eta} \in 	\mathbb{R}_0^+$, where $D$ denotes the diagonal. 
We refer to this as a \emph{generalized rectangular array} since different shapes are obtained depending on the value of $ {\eta}$. If we fix the array's aperture area $A_{\rm array} = NA$, the diagonal of the array (i.e., the aperture length) becomes $  \sqrt{A_{\rm array}(1+ {\eta}^2)/ {\eta}} $ which varies with $ {\eta}$. If we instead fix  the array's aperture length $D_{\rm array} = 
 \sqrt{N}D$, the array's aperture area $   {\eta} (D_{\rm array} )^2/(1+ {\eta}^2)$ will depend on $ {\eta}$. We consider both scenarios in this paper to study how the shape of the beamforming depends on the array geometry. We take the viewpoint of a receiver array with an isotropic transmitter, for ease in presentation, but the same results apply in the reciprocal setup with a transmitting array.

The antenna element with index $(n,m)$ is located at the point $(x_n,y_m,0)$, where 
\begin{equation}\label{eq_III_AntennaElementCoordinate}
    x_n = \left( n - \frac{\sqrt{N}+1}{2} \right)   w ,\quad y_m = \left( m - \frac{\sqrt{N}+1}{2}\right)  \ell,   
\end{equation}
and covers an area of $\mathcal{A} = \left\{ (x,y,0): |x-x_n| \leq \frac{w}{2},  |y-y_m| \leq \frac{\ell}{2} \right\}$.

Let us now consider the normalized array gain in \eqref{eq_II_NormalizedArrayGain}. We can either numerically evaluate it or use a Fresnel approximation that allows us to obtain an analytical expression. In the latter case, we approximate $  E(x,y) $ in \eqref{eq_II_ElectField} as
\begin{equation}\label{eq_III_FresnelApprox}
    E(x,y) \approx \frac{E_0}{\sqrt{4 \pi} z} e^{-j \frac{2 \pi}{\lambda} \left( z+\frac{x^2}{2z} + \frac{y^2}{2z} \right)}.
\end{equation}
 This approximation is tight in the part of the Fresnel region where the amplitude variations are negligible over the aperture. The term $z+\frac{x^2}{2z} + \frac{y^2}{2z} $ determines the phase variations in \eqref{eq_III_FresnelApprox} and is obtained from a first-order Taylor approximation of the Euclidean distance between transmitter and array:
 $\sqrt{x^2+y^2+z^2} = z\sqrt{1+ \left(\frac{x^2+y^2}{z^2}\right)} \approx  z + \frac{x^2+y^2}{2z}$. 
 The approximation error is small when  $\frac{x^2+y^2}{z^2} \leq 0.1745$, which results in relative errors below $3.5 \cdot 10^{-3}$ \cite{2021_Björnson_Asilomar}. 
To give a concrete example, we consider a square array with $ {\eta}=1$, $N=10^4$ antennas, and $D=0.025$ meter. The relative error is  smaller than $3.5 \cdot 10^{-3}$ if $z > 1.2 D \sqrt{N} $.  

Suppose the matched filtering in the array is focused on the point $(0, 0, F)$, which might be different from the location $(0,0,z)$  of the transmitter. We will now use the Fresnel approximation in \eqref{eq_III_FresnelApprox} to compute the normalized array gain by injecting the phase-shift $e^{+j \frac{2\pi}{\lambda} \left(\frac{x^2}{2F} + \frac{y^2}{2F}\right)}$ into the integrals in \eqref{eq_II_NormalizedArrayGain} to represent a continuous matched filter \cite{Hu2018a,Dardari2020a,2020_Björnson_JCommSoc}. 
 
\begin{Theorem}\label{Fresnel_Approx_Rect}
When the transmitter is located at $(0,0,z)$ and the matched filtering is focused on $(0, 0, F)$, the  Fresnel approximation of the normalized array gain for the generalized rectangular array becomes
\begin{equation}\label{eq_III_ApproxGainRect}
    \hat{G}_{\rm rect}( {\eta}) = \frac{ ( {{C}^{2}}\left(  {\eta}\sqrt{ a } )+{{S}^{2}}(  {\eta}\sqrt{ a } \right) )  ( {{C}^{2}}( \sqrt{ a } )+{{S}^{2}}( \sqrt{ a } ) ) }{( {\eta}  a  )^2},
\end{equation}
 where $C(\cdot )$ and $S(\cdot )$ are the Fresnel integrals \cite{1956_Polk_TAP}, $ a = \frac{d_{FA}}{4{z}_{\rm eff}(1+{{ {\eta}}^{2}})}$, and $z_{\rm eff} = \frac{Fz}{|F-z|} $.
\end{Theorem}
\begin{proof}
The proof is given in Appendix~\ref{App_Fresnel_Approx}
and is inspired by the derivation in \cite[Eq. (22)]{1956_Polk_TAP}.
\end{proof}

In Fig.~\ref{Fig_Fresnel_vs_Exact_Rectangular}, we evaluate the tightness of the analytical expression provided in Theorem~\ref{Fresnel_Approx_Rect}.  We can see that the Fresnel approximation is close to the exact normalized array gain, especially when $z\geq d_B$ where $d_B=2 D \sqrt{N} = 400 d_F $ denotes the Bj{\"o}rnson distance \cite{2021_Björnson_Asilomar} since then the amplitude variations are insignificant.  Since $D$ depends on $\lambda$, it also depends on the carrier frequency $f_c$. We consider $f_c=3$ GHz in this paper. The exact normalized antenna gain is computed numerically using \eqref{eq_II_NormalizedArrayGain} and  $E(x,y)$ as stated in \eqref{eq_II_ElectField} with the injected phase-shift  $e^{+j \frac{2\pi}{\lambda} (\frac{x^2}{2F} + \frac{y^2}{2F})}$ that represents the matched filtering.

\begin{figure}
    \centering
    \includegraphics[width=\textwidth]{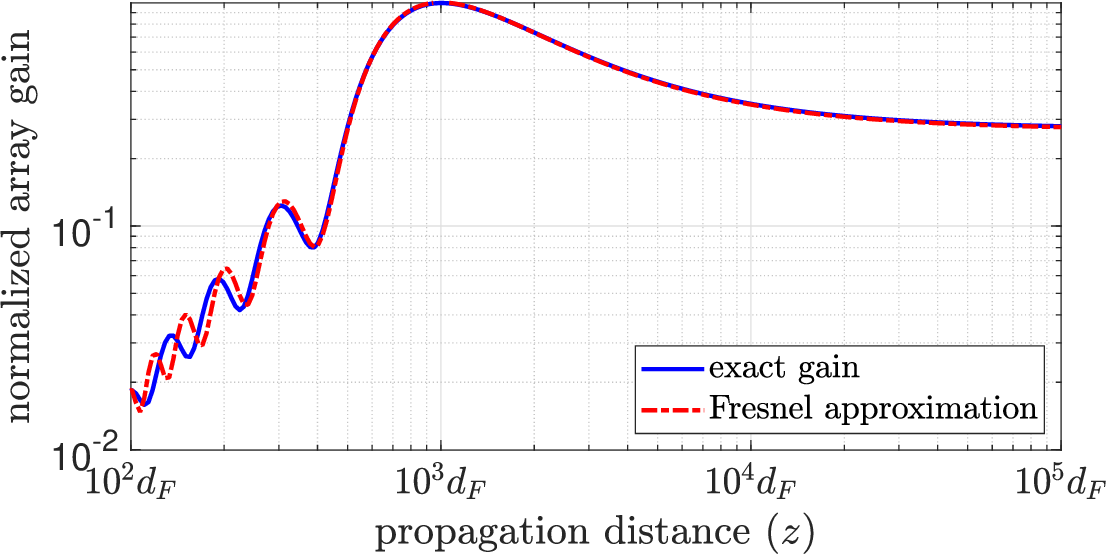}
    \caption{The analytic approximation of the normalized array gain in Theorem~\ref{Fresnel_Approx_Rect} 
    for a rectangular array with a broadside transmit signal and $F = 1000 d_F$.  We set $f_c=3$ GHz,  {diagonal of antenna $\lambda/4$},  $N=10^4$, and $ {\eta}=4$.  }
    \label{Fig_Fresnel_vs_Exact_Rectangular}
\end{figure}

The array gain is the same irrespective of how the rectangular array is rotated in the $xy$-plane, which is proved as follows.

\begin{Corollary}\label{corr_symmetric_function}
The array gain expression in \eqref{eq_III_ApproxGainRect} is a symmetric function in the sense that $\hat{G}_{\rm rect}( {\eta}) = \hat{G}_{\rm rect}( {\eta}') $, where $ {\eta}'= 1/ {\eta}$.
\end{Corollary}
\begin{proof}
Let us define $k \triangleq \frac{d_{\rm FA}}{4 z_{\rm eff}}$ and  $a' =\frac{d_{\rm FA}}{4 z_{\rm eff}(1+(1/ {\eta})^2)}$. One can express $ {\eta}\sqrt{ a } = \sqrt{\frac{kc^2}{1+ {\eta}^2}}$. If we replace $ {\eta}$ by $ {\eta}'$, we obtain $ {\eta}'\sqrt{ a' } = \sqrt{\frac{k(1/ {\eta})^2}{1+(1/ {\eta})^2}} = \sqrt{\frac{k}{1+ {\eta}^2}} =  a $ and $\sqrt{ a' } =\sqrt{\frac{k}{1+(1/ {\eta})^2}} = \sqrt{\frac{kc^2}{1+ {\eta}^2}} =  {\eta}\sqrt{ a } $. Therefore, the numerator of \eqref{eq_III_ApproxGainRect} can also be written as  $  ( {{C}^{2}}(  {\eta}'\sqrt{ a' } )+{{S}^{2}}(  {\eta}'\sqrt{ a' } ) )$ $( {{C}^{2}}( \sqrt{ a' } )+{{S}^{2}}( \sqrt{ a' } ) ) $. We notice that the denominator of \eqref{eq_III_ApproxGainRect} satisfies $( {\eta} a )^2 = ( {\eta}' a' )^2 $. Hence, $\hat{G}_{\rm rect}( {\eta}) = \hat{G}_{\rm rect}( {\eta}') $.
\end{proof}

When applying matched filtering in the radiative near-field, the beamforming might have a limited BD; in fact, \cite{2021_Björnson_Asilomar} advocates that the near-field ends at the point where the BD becomes infinite and not at the Fraunhofer distance. For a given focus point $F$, this implies that the array gain is only large for a limited range of transmitter distances $z$. This stands in contrast to far-field focusing where the array gain remains large in the range $z \in [F,\infty)$.
We consider the $3$ dB BD, ${\rm BD}_{3 {\rm dB}}$,  where the normalized array gain is higher than half of its peak gain. The 3 dB BD definition aligns with $3$ dB beam width concept and has been used in \cite{2021_Björnson_Asilomar,1962_Sherman_TAP,2017_Nepa_APMag}.
We will compute an analytical approximation of the BD using Theorem~\ref{Fresnel_Approx_Rect}. As the peak value of $\hat{G}_{\rm rect}( {\eta})$ is $1$, we are interested in the value of $a$ for which $\hat{G}_{\rm rect}( {\eta})\approx0.5$, indicating the $3$ dB loss. We further denote this value of $a$ as $a_{3{\rm dB}}$, which can be obtain numerically as follows. Initially, we calculate $\hat{G}_{\rm rect}( {\eta})$ using the closed-form solution with varying values of the parameter $a$. Subsequently, we treat this computation as a look-up table, enabling us to identify the specific $a$ value that yields $\hat{G}_{\rm rect}( {\eta}) \approx 0.5$. To enhance precision, we can regenerate $\hat{G}_{\rm rect}( {\eta})$ with finer increments in the $a$ values. The $a$ value that results in the closest $\hat{G}_{\rm rect}( {\eta})$ approximation to $0.5$ is designated as $a_{3 \rm{dB}}$. The $3$ dB BD is computed as follows by using the Fresnel approximation.

\begin{Theorem}\label{BD_analytical_rectangular}
The $3$ dB BD of a rectangular array that is focused on $F\geq d_B$ is approximately computed as
\begin{equation}\label{eq_III_Analytical_BD_Rect}
 {\rm BD}_{3 {\rm dB}}^{\rm rect}  = 
        \begin{cases}
        \frac{8 d_{FA} F^2 a_{3{\rm dB}}  (1+ {\eta}^2) }{ d_{FA}^2 - ( 4 F a_{3{\rm dB}} (1+ {\eta}^2) )^2}  , & F < \frac{d_{FA}}{4 a_{3{\rm dB}} (1+ {\eta}^2)},\\
        \infty, & F \geq  \frac{d_{FA}}{4 a_{3{\rm dB}} (1+ {\eta}^2)}.
        \end{cases}
\end{equation}
\end{Theorem}
\begin{proof}
 We can obtain a value of $a$ from \eqref{eq_III_ApproxGainRect}, for which $\hat{G}_{\rm rect}( {\eta})$ is approximately $0.5$. The value of $a$ is denoted as $a_{3{\rm dB}}$. From Theorem~\ref{Fresnel_Approx_Rect}, we know that $a_{3{\rm dB}} = \frac{d_{FA}}{ 4z_{\rm eff} (1+ {\eta}^2)} =  \frac{d_{FA} |F-z|}{4 Fz (1+ {\eta}^2)} $, where the value of $z$ indicates the $3$ dB gain in the propagation distance. We can then write  $z = \frac{d_{FA}F}{d_{FA} \pm 4 F a_{3{\rm dB}}  (1+ {\eta}^2) }$. There are two $z$-values resulting in the $3$ dB gain. The difference between those two is the range of distances where the gain is greater than half of its peak gain, referred to as the $3$ dB BD. Therefore, the $3$ dB BD is computed as   $\frac{d_{FA} F }{d_{FA} - 4 F a_{3{\rm dB}} (1+ {\eta}^2) } - \frac{d_{FA} F }{d_{FA} + 4 F a_{3{\rm dB}} (1+ {\eta}^2) } $ which results in \eqref{eq_III_Analytical_BD_Rect}. When $F \geq  \frac{d_{FA}}{4 a_{3{\rm dB}} (1+ {\eta}^2)}$, one of the $z$-values is negative which implies there is no upper limit on the beamwidth. Notice that ${\rm BD}_{3 {\rm dB}}^{\rm rect} $ does not depend on the carrier frequency since the wavelength term are canceled out when substituting  $d_{FA} = \frac{2 D^2 N}{\lambda}$ into \eqref{eq_III_Analytical_BD_Rect}.
\end{proof}

 {Referring to \eqref{eq_III_Analytical_BD_Rect}, it becomes evident that distinct beam depth intervals emerge when directing focus to distances within $d_{\mathrm{F}}/(8a_{3\mathrm{dB}})$ for $\eta =1$. In particular, when dealing with a square ELAA where $M=N$, we can utilize the range of $3$ dB BD $\left[\frac{d_{FA} F }{d_{FA} + 4 F a_{3{\rm dB}} (1+ {\eta}^2) } , \frac{d_{FA} F }{d_{FA} - 4 F a_{3{\rm dB}} (1+ {\eta}^2) } \right] $ to compute five distinct focal points:  $F = d_{FA}/20$, $F = d_{FA}/40$, $F = d_{FA}/60$, $F = d_{FA}/80$, and $F = d_B = d_{FA}/100$ that give us non-overlapping $3$ dB BD intervals. Fig. \ref{fig_multiplex_dist} depicts the normalized array gain when an ELAA with $M=N=200$ and $D=\lambda/2$ focuses on the mentioned five points using matched filtering. 
\begin{figure}
    \centering
    \includegraphics[width=\textwidth]{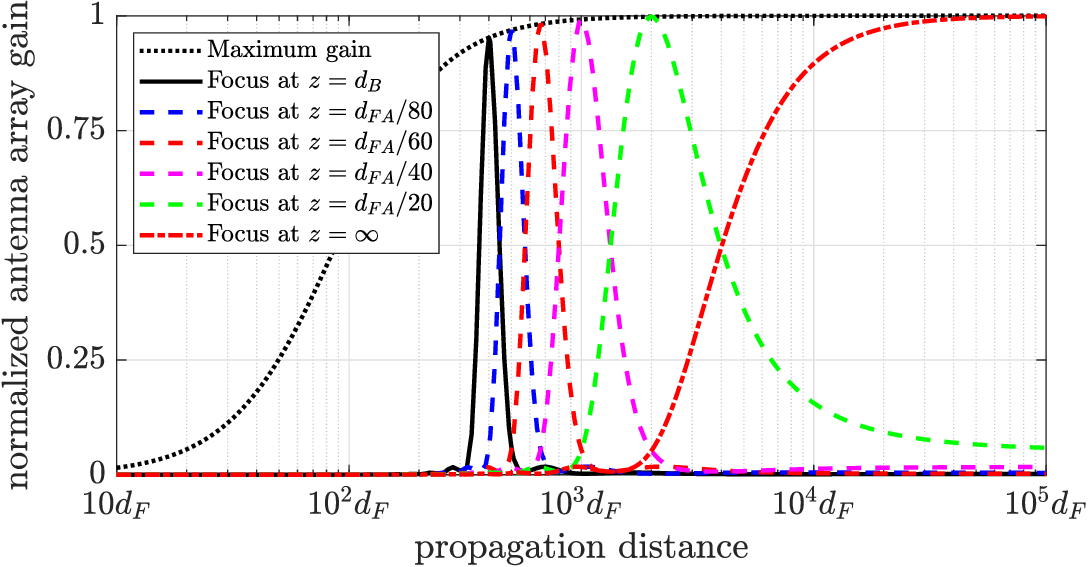}
    \caption{ {The normalized array gains are illustrated for several focal points in the same angular direction, each exhibiting non-overlapping 3 dB BD.}}
    \label{fig_multiplex_dist}
\end{figure}
It can be seen that the $3\,$dB beam depth intervals are non-overlapping, which lets the ELAA multiplexes five users in the near-field with limited interference.}

We will now exemplify how we can multiplex users with limited interference based on the non-overlapping $3$ dB BD  focal points in  Fig.~\ref{fig_multiplex_dist}. We evaluate the performance in terms of the achievable sum rate. We consider a downlink scenario, where  $K$ single-antenna users are served by a large base station array. The channel vector to user $k$ is denoted as $\vect{h}_k \in \mathbb{C}^{N}$.
Each element of this vector represents the channel between user $k$ and an antenna element indexed by $(n, m)$, which is given in \eqref{eq_II_ArrayChResponse}.
 The received  signal $y_k \in \mathbb{C}$ can be expressed as
\begin{equation} \label{eq:MU-MIMO-userk1}
    y_k = \vect{h}_k^{\Htran} \sum_{i=1}^{K} \vect{w}_i x_i + n_k,
\end{equation}
where $x_i \in \mathbb{C}$ is the data signal intended for user $i$, $\vect{w}_i \in \mathbb{C}^{N}$ is the corresponding precoding vector, and $n_k  \in \mathbb{C}$ is the noise at user $k$.
The combined received signals of all users is
\begin{equation} \label{eq:MU-MIMO-userk}
    \underbrace{\begin{bmatrix}
        y_1 \\
        \vdots \\
        y_K
    \end{bmatrix}}_{\vect{y}} = \underbrace{\begin{bmatrix}
        \vect{h}_1^{\Htran} \\ \vdots \\ \vect{h}_K^{\Htran}
    \end{bmatrix}}_{\vect{H}^{\Htran}}  
    \underbrace{\begin{bmatrix}
        \vect{w}_1 & \ldots & \vect{w}_K
    \end{bmatrix}}_{\vect{W}}
    \underbrace{\begin{bmatrix}
        x_1 \\
        \vdots \\
        x_K
    \end{bmatrix}}_{\vect{x}} +     \underbrace{\begin{bmatrix}
        n_1 \\
        \vdots \\
        n_K
    \end{bmatrix}}_{\vect{n}},
\end{equation}
which can be expressed as $\vect{y} = \vect{H}^{\Htran} \vect{W} \vect{x} + \vect{n}$.
The transmitter can then suppress interference by using the minimum mean-square-error (MMSE) precoding   $  \vect{W} = \alpha \vect{H} (\vect{H}^{\Htran} \vect{H}  + \vect{ I})^{-1}$, where $\alpha = \frac{1}{\sqrt{\sum_{i=1}^{N} \sum_{j=1}^{K} |h_{ij}|^2}} $ is a power normalization factor.
We can compute the achievable sum rate as
\begin{equation}
    R_{\rm sum} = \sum_{k=1}^K \log_2 \left(1 + \frac{p_k |\vect{h}^{\Htran}_k \vect{w}_k|^2}{ \sum_{\substack{j=1 \\ j \neq k}}^{j=K} p_j   | \vect{h}^{\Htran}_k \vect{w}_j |^2 + 1 } \right),
\end{equation}
where $p_k$ is the transmit power of user $k$.
We will now compare the sum rate obtained by locating users at the points where there exists non-overlapping  $3$ dB BD (see Fig. \ref{fig_multiplex_dist}), referred to as the $3$ dB BD separation, i.e.,  $d_{FA}/20, d_{FA}/40, d_{FA}/60, d_{FA}/80$, and $d_B = d_{FA}/100$ with the sum rate achieved by distributing the users uniformly within the range of finite depth beamforming $[d_B,d_{FA}/10]$. 
Fig. \ref{fig_sumRate_snr_BW} demonstrates that the one with  $3$ dB BD separation results in a higher sum rate when the SNR is sufficiently high, e.g., $20$ dB. The SNR is defined as the ratio of transmitter power over the noise power.
\begin{figure}
\centering
{\includegraphics[width=\textwidth]{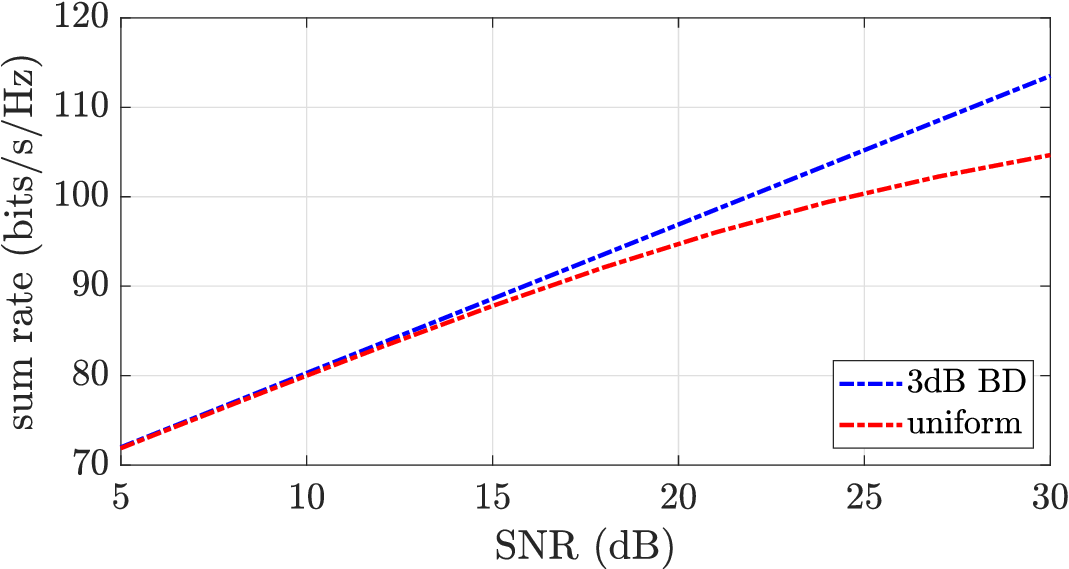}}\hfill
\caption{ {The comparison of sum rate with respect to SNR with different focal points for multiplexing $5$ users within the range of finite depth beamforming $[d_B,d_{FA}/10]$.}}
\label{fig_sumRate_snr_BW}
\end{figure}
We extend our analysis to evaluate the mean of achievable sum rate across  $100000$ random user locations with respect to the number of users with an SNR of $25$ dB. The users are randomly placed within the range of  $[d_B,d_{FA}/10]$. The outcome is illustrated in Fig. \ref{figure_sum_rate_users}. Notably, the highest average sum rate is realized when multiplexing five users, which is what the beam depth analysis suggests. 
\begin{figure}
    \centering
    \includegraphics[width=0.95\textwidth]{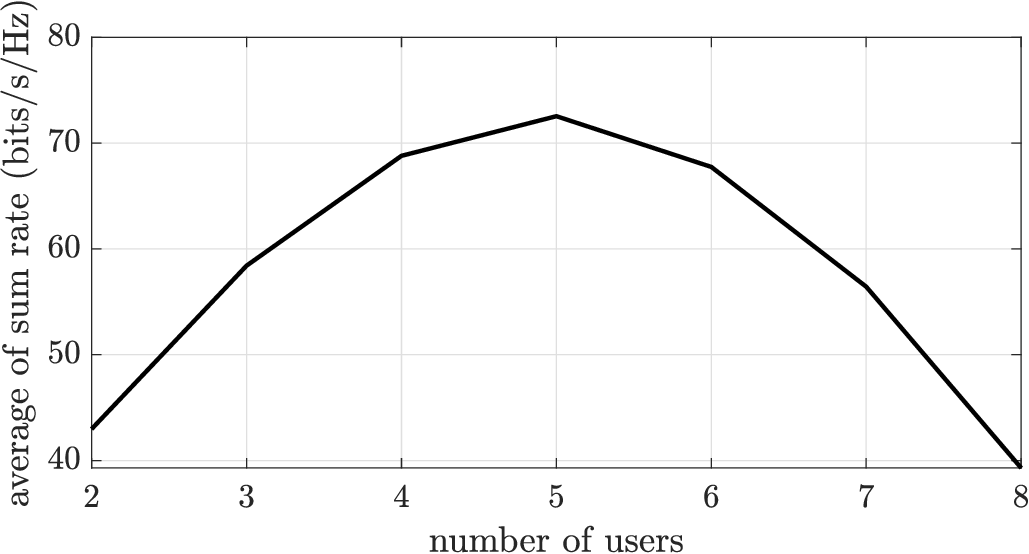}
    \caption{ {The mean of achievable sum rate with respect to the number of users, where users are placed randomly within the range of $[d_B,d_{FA}/10]$ associated with the range of finite beamforming.}}
    \label{figure_sum_rate_users}
\end{figure}

Looking more closely at the simulation results, when multiplexing $5$ users and considering $10000$ random user locations, the highest observed sum rate was $105.16$. This occurred when the users' locations closely matched the $3$ dB BD separation rule. The sum rate was close to the sum rate of $105.20$ (see Fig.~\ref{fig_sumRate_snr_BW}), obtained when the users are positioned using the $3$ dB BD separation with an SNR of $25$ dB. Hence, the $3$ dB BD separation serves as the ideal upper bound for user multiplexing in the distance domain.

The fact that the BD is finite in the radiative near-field enables spatial multiplexing of closely spaced users, even in the same angular direction. It is essential to understand which array geometry makes this feature most prevalent. ${\rm BD}_{3 {\rm dB}}^{\rm rect}$ in Theorem~\ref{BD_analytical_rectangular} is maximized when $ {\eta}=1$. The $3$ dB BD depends on the value of $a_{3{\rm dB}}$ for which $\hat{G}_{\rm rect}( {\eta}) = 0.5$. Since it is difficult to get a closed-form solution of $a_{3{\rm dB}}$ with respect to $ {\eta}$, we simulate it as given in Fig.~\ref{Fig_x_3dB_mult}. 
\begin{figure}
\centering
\includegraphics[width=\textwidth]{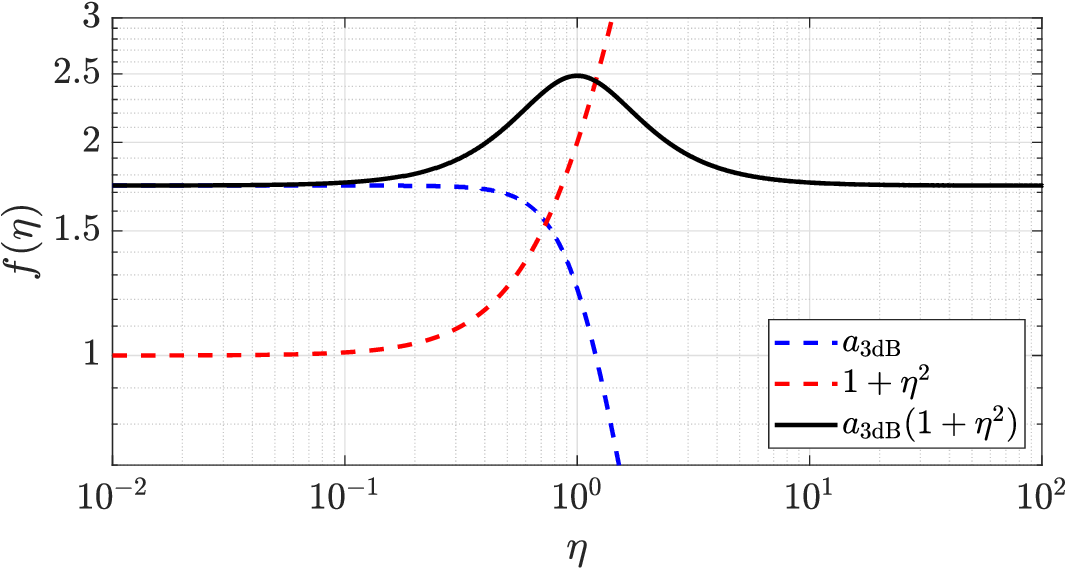}
\caption{The multiplication of $a_{3{\rm dB}}$ and $(1+ {\eta}^2)$ with respect to $ {\eta}$. We set $f_c=3$ GHz, $N=10^4$, and fixed the aperture area $A=\lambda^2/ 32$. The multiplication is maximized when $ {\eta}=1$.}
\label{Fig_x_3dB_mult}
\end{figure}
$a_{3{\rm dB}}$ is a decreasing function while $(1+ {\eta}^2)$ is an increasing function. Multiplying them together results in a function maximized when $ {\eta}=1$. Since the numerator and denominator in ${\rm BD}_{3 {\rm dB}}^{\rm rect}$ are maximized when the multiplication result of $a_{3{\rm dB}}$ and $1+ {\eta}^2$ is maximum, the largest value of  ${\rm BD}_{3 {\rm dB}}^{\rm rect}$ is obtained when $ {\eta}=1$.

We obtain a square array when $ {\eta}=1$ and then  ${\rm BD}_{3 {\rm dB}}^{\rm rect} = \frac{20 d_{FA} F^2  }{d_{FA}^2 - 100 F^2}$,  which is exactly the expression of the BD in \cite[Eq. (23)]{2021_Björnson_Asilomar}.  Another property from the analytical $3$dB BD in \eqref{eq_III_Analytical_BD_Rect} is that it goes to infinity when the focus distance $F$ is greater than $\frac{d_{FA}}{4 a_{3{\rm dB}} (1+ {\eta}^2)}$. We refer to this focus boundary as the \emph{finite BD limit} and stress that it separates the near-field and far-field propagation characteristics. To study this limit, we plot its value in Fig.~\ref{Fig_focus_lim} as a function of $ {\eta}$. 
The limit does not depend on the carrier frequency.
When doing so, there are two options.
\begin{figure}
    \centering
    \includegraphics[width=\textwidth]{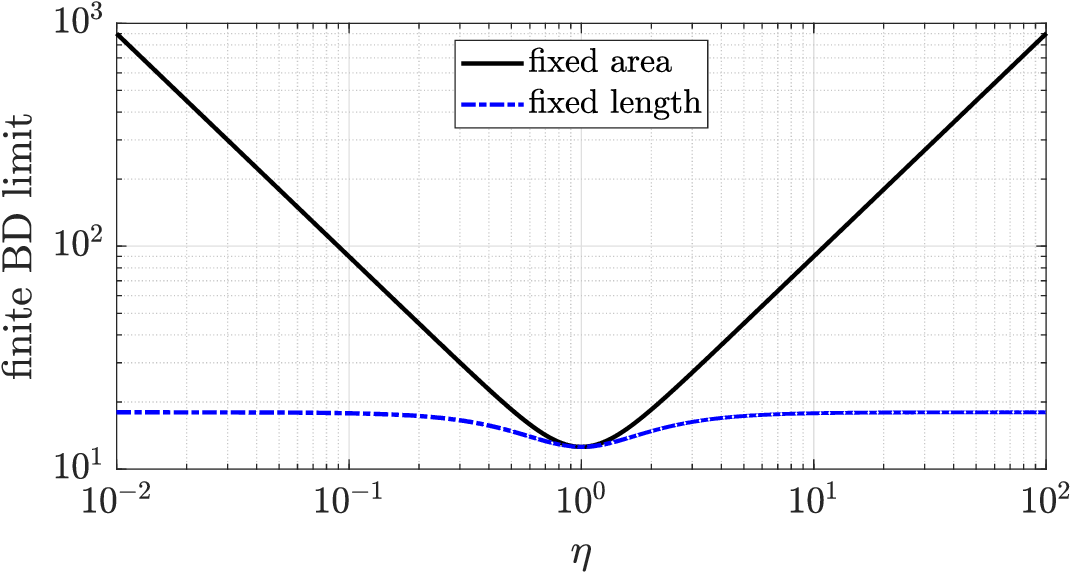}
    \caption{The finite BD limit with respect to $ {\eta}$. The limit indicates the boundary for the distance to the focus point $F$ where the BD becomes infinity.
    We set $N=10^4$. The finite BD limit is minimized when $ {\eta}=1$.}
    \label{Fig_focus_lim}
\end{figure}
We can either fix the aperture length or the total aperture area to observe how  the finite BD limit depends on $ {\eta}$. If we fix the aperture area as $A_{\rm array} =N \lambda^2/32 $, the aperture length varies with respect to  $ {\eta}$ such that $ \sqrt{A_{\rm array} (1+ {\eta}^2)/ {\eta}}$. Since $d_{FA} = N d_F$, where $d_F = \frac{2D^2}{\lambda}$, the finite BD limit $\frac{d_{FA}}{4 a_{3{\rm dB}} (1+ {\eta}^2)}$ highly depends on the changes of the aperture length $D_{\rm array}$, as demonstrated in Fig.~\ref{Fig_focus_lim}. The finite BD limit is minimized when $ {\eta}=1$. 
If we fix the aperture length $D_{\rm array} = \sqrt{N}\lambda/4$, the limit now depends on $a_{3{\rm dB}} (1+ {\eta}^2)$. More specifically, they are inversely proportional. The finite BD limit  is still minimized when $ {\eta}=1$. We further evaluate the achievable sum rate with respect to the width-to-height ratio $\eta$, as depicted in Fig.~\ref{fig_sum_rate_eta}. It becomes evident that taller arrays ($\eta \ll 1$) and wider arrays ($\eta \gg 1$) exhibit an improved sum rate when compared to that of square arrays. This observation highlights that the square array yields a greater beam depth, consequently resulting in a diminished sum rate. 
\begin{figure}
    \centering
    \includegraphics[width=\textwidth]{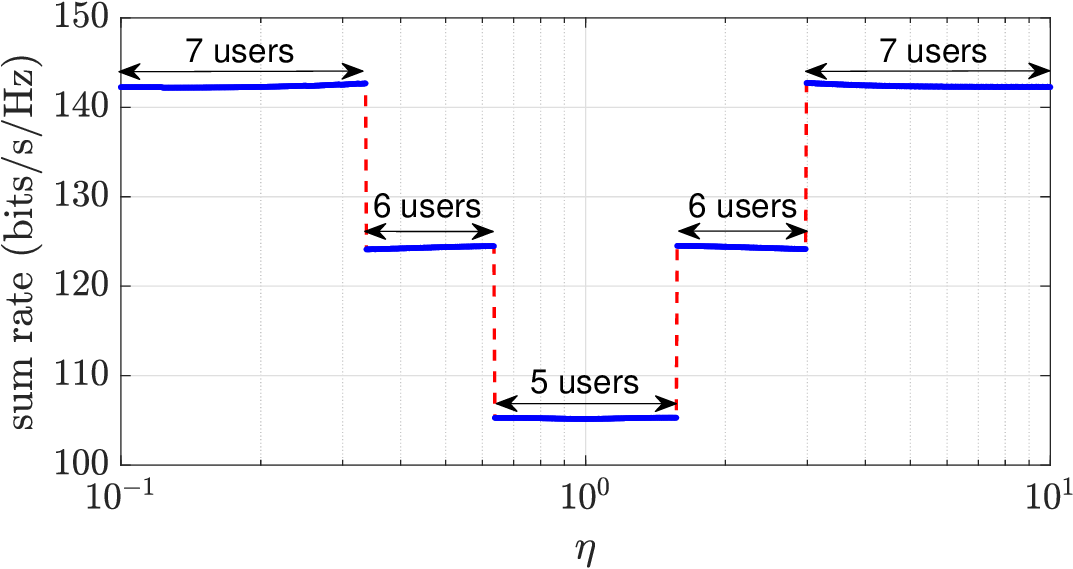}
    \caption{ {The achievable sum rate with respect to the width-to-height ratio $\eta$ for non-overlapping $3$dB BD users' separation.}}
    \label{fig_sum_rate_eta}
\end{figure}
We summarize these important results as follows.

\begin{Observation}\label{Remark_DoF_Rect}
 {The finite BD limit mostly depends on the aperture length $D_{\rm array}$. Finite-depth beamforming is attainable at distances greater than or equal to the Bj\"ornson distance and lower than or equal to the finite BD limit. The finite BD limit is lower than the Fraunhofer distance, for example, the finite BD limit for a square array is an order magnitude lower than the Fraunhofer array distance. The finite BD regime is largest for a uniform linear array (ULA).}
\end{Observation}

In the following, we numerically compute the $3$ dB BD of the rectangular array with respect to $ {\eta}$. We consider both fixed area and length scenarios.

\subsection{Fixed Array Aperture Area vs. Fixed Array Aperture Length}
\label{S_Rect_length_vs_area}

When varying the shape of the antenna elements through $ {\eta}$, we can fix the array's aperture area $A_{\rm array} = {N} \lambda^2/32$, while letting the array's aperture length $  D_{\rm array} = A_{\rm array}(1+ {\eta}^2)/ {\eta} $ vary. Alternatively, we can fix the aperture length $D_{\rm array} = \sqrt{N}\lambda/4$ so that the aperture area depends on $ {\eta}$ as $ {\eta} (D_{\rm array})^2/(1+ {\eta}^2)$. We will analyze both cases below.

Fig.~\ref{Fig_Beamdepth_vs_c} depicts the $3$ dB BD with respect to $ {\eta}$ for a fixed area. The normalized array gain is computed numerically using \eqref{eq_II_NormalizedArrayGain}. The $3$ dB BD is maximum when the array is square-shaped ($ {\eta}=1$), i.e., around $400 d_F$. The BD pattern is symmetric for tall ($ {\eta}<1$) and wide ($ {\eta}>1$)  arrays, which is aligned with Corollary~\ref{corr_symmetric_function}. Note that as $ {\eta}$ attains a smaller/bigger value, the two-dimensional UPA gradually approaches the shape of a one-dimensional ULA. In fact, the ULA can be regarded as a special case of our generalized rectangular array model that is achieved when the height approaches the aperture length: $\lim_{ {\eta} \to 0} \frac{D_{\rm array}}{\sqrt{1+ {\eta}^2}} = D_{\rm array}$. The smallest BD is obtained in Fig.~\ref{Fig_Beamdepth_vs_c} when the array approaches a ULA with around $50 d_F$ when $ {\eta}=0.1$ and $ {\eta}=10$. 

\begin{figure}
    \centering
    \includegraphics[width=\textwidth]{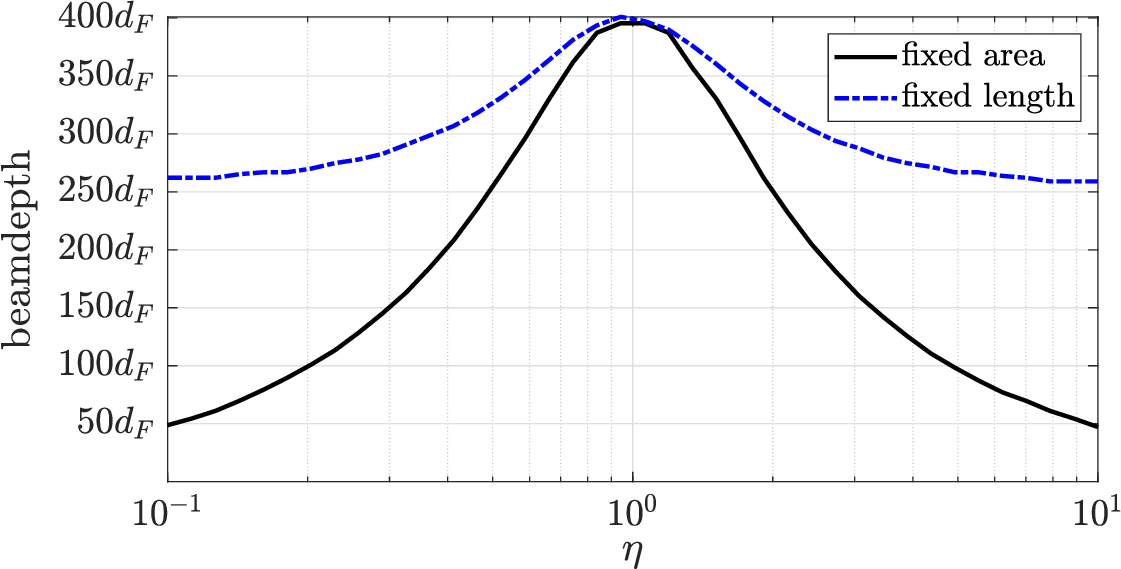}
    \caption{The BD with respect to rectangular array shapes. We set $f_c=3$ GHz, $N=10^4$. The BD pattern is symmetric for tall ($ {\eta}<1$) and wide ($ {\eta}>1$)  arrays, and its largest point is when $ {\eta}=1$.}
    \label{Fig_Beamdepth_vs_c}
\end{figure}
 
Suppose we instead fix the array's aperture length so that the aperture area depends on $ {\eta}$. The $3$ dB BD for this case is also plotted in Fig.~\ref{Fig_Beamdepth_vs_c} with respect to $ {\eta}$. Similarly to the results when we fix the aperture area, the $3$ dB BD is maximum when the array is square-shaped ($ {\eta}=1$), i.e., around $400 d_F$; The BD pattern is symmetric for tall and wide arrays, and the smallest BD is obtained when the array approaches a ULA with around $250 d_F$, when $ {\eta}=0.1$ and $ {\eta}=10$. It is worth noting that the gap between the highest and smallest   $3$ dB BD over $ {\eta}$ for a fixed aperture length is much lower than the one in the fixed area scenario. This implies that the array's aperture length has a more significant impact on the $3$ dB BD than the array area. 

In the following, we discuss the practical considerations for the two considered scenarios:
\begin{itemize}
    \item  \textbf{Fixed array aperture area:} One of the main limiting factors when deploying antenna arrays in practice is the wind load. The total wind load is directly proportional to the aperture area, which is why we consider it here. The frontal and lateral wind loads also depend on the shape of the array \cite{2019_Huawei_whitepaper,2019_antenna_standard}. A tall array has a smaller wind-exposed frontal area compared to a wide array, while a wide array has a smaller wind-exposed lateral area than a tall array. Hence, the array's height poses a constraint on the stability, while the array's width acts as a limiting factor against drag wind forces.
    \item  \textbf{Fixed array aperture length:} If the antenna array is deployed on a wall/structure with a predefined size, then the primary practical concern is the maximum aperture length that must fit the largest dimension of the structure. Since the length is measured diagonally, one can possibly obtain a larger aperture area by utilizing two deployment dimensions.
\end{itemize}

\begin{Observation}
The largest $3$ dB BD is obtained with a square array ($ {\eta}=1$). For tall ($ {\eta}<1$) and wide ($ {\eta}>1$)  arrays, the BD patterns are symmetric. The smallest BD is obtained when the array approaches a ULA. The array's aperture length has a more significant impact on the $3$ dB BD  variation than  the array's aperture area.
\end{Observation}

\section{Gain and Beam Depth of Rectangular Arrays With a Non-broadside Transmitter}
\label{Sect_BD_Rect_NonBroadside}

\begin{figure}
    \centering
    \includegraphics[width=0.55\textwidth]{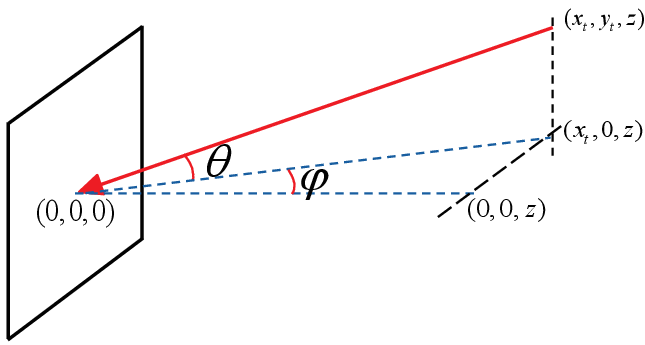}
    \caption{The non-broadside transmitter with a rectangular array at the receiver. The transmitter is located at the distance $d$ with an azimuth angle of $\varphi$  {and an elevation angle of $\theta$} from the center of the receiver array.}
    \label{Fig_non_broadside}
\end{figure}

In this section, we consider rectangular UPAs with a transmitter located in a non-broadside direction. Consequently, there will be reductions in the effective area due to directivity, as well as polarization losses.  When denoting the non-zero azimuth angle as $\varphi$, the location coordinate of the transmitter is  $(x_t,y_t,z)$, where $x_t=d {\sf sin}(\varphi) {\sf cos}(\theta)$, $y_t=d {\sf sin}(\theta) $, and $z = d{\sf cos}(\theta)  {\sf cos}(\varphi)$, and $d$ represents the distance of the transmitter to the array's center, as shown in Fig.~\ref{Fig_non_broadside}. 
The analysis is carried out by considering a non-zero azimuth angle, but the same methodology can also be applied to consider a non-zero elevation angle.
We utilized the rectangular array design discussed in the preceding section.

We first evaluate the first-order Taylor approximation of the Euclidean distance between the transmitter and receiver when the transmitter is off the center of the array's broadside direction. The Euclidean distance determines the phase variations in each antenna element of the array. Therefore, an accurate approximation of the Euclidean distance results in an accurate approximation of the array gain. The Euclidean distance between the transmitter and an antenna element located at $(x_n,y_m,0)$ is calculated as $f_{n,m}(z) = \sqrt{(x_{n}-x_t)^2 + y_m^2 +z^2} = z\sqrt{1+ \frac{(x_n-x_t)^2 + y_m^2}{z^2} } $. The first-order Taylor approximation of the Euclidean distance is $\tilde{f}_{n,m}(z) = z \left(1 + \frac{( x_n  -x_t)^2 +y_m^2}{2z} \right)$. This approximation is tight when $\frac{(x_n-x_t)^2 + y_m^2}{z^2} $ is small (e.g., $\leq 0.1745$) as mentioned in Section~\ref{Sect_Beamdepth_Rect_Broadside}. However, since $z= d \cos(\varphi)$, the approximation will be significantly imprecise when $\varphi$ approaches $\pm \pi/2$. We can mathematically express the above case as $\lim_{\varphi \rightarrow \pm \frac{\pi}{2}} \frac{\left(x_n-x_t\right)^2 + y_m^2}{\left(d \cos(\varphi)\right)^2} = \infty$. To evaluate this, we first calculate the mean absolute error of the antenna elements in the array as  $\epsilon = \frac{1}{N} \cdot \sum_{n=1}^{\sqrt{N}} \sum_{m=1}^{\sqrt{N}}\left| f_{n,m}(z)- \tilde{f}_{n,m}(z)\right|$, where the subscripts $n,m \in \{1,\ldots,\sqrt{N}\}$ are the location indices of the antennas and  $z = d_B$. Then, we plot $\epsilon$ with respect to $\varphi$ demonstrated by the blue curve in Fig.~\ref{Fig_Fresnel_approx_theta}, which we refer to as the \emph{direct approximation}. The curve implies that the approximation is below the threshold $3.5 \cdot 10^{-3}$ for $-\pi/16 \leq \varphi \leq \pi/16$. As $\varphi$ approaches $\pm \pi/8$, the approximation error exhibits a steep incline. We obtain a better approximation by first manipulating the expression of the Euclidean distance between the transmitter and receiver as
$z\sqrt{1+ \frac{\left(x-x_t\right)^2 + y^2}{z^2} } = z \sqrt{\underbrace{1 + \tan^2(\varphi)}_{1/\cos^2(\varphi)} + \frac{x^2 + y^2}{\left( d \cos(\varphi)\right)^2} - \frac{2x \tan(\varphi)}{d \cos(\varphi)} } =$ 
$ \underbrace{d\sqrt{1 + \left(\frac{x^2+y^2-2xd\sin(\varphi)}{d^2}\right)}}_{ f_{n,m}(d)}$.
Then, we approximate it using the first-order Taylor approximation as 
\begin{equation}\label{eq_IV_FresnelApp}
   \tilde{f}_{n,m}(d) =  d + \frac{x^2 + y^2 -2xd \sin(\varphi)}{2d}.
\end{equation}
Now, the approximation is tight when $\frac{x^2+y^2-2xd\sin(\varphi)}{d^2}$ is small and there is no $\cos(\varphi)$-term in the denominator in contrast to $\tilde{f}_{n,m}(z)$. We refer to this approximation as the \emph{indirect approximation}, since we manipulate the Euclidean distance expression before then applying the Taylor approximation. The corresponding mean absolute error is $\epsilon = \frac{1}{N} \cdot \sum_{n=1}^{\sqrt{N}} \sum_{m=1}^{\sqrt{N}}\left| f_{n,m}(d)- \tilde{f}_{n,m}(d)\right|$. The new approximation is evaluated in Fig.~\ref{Fig_Fresnel_approx_theta}, where it shows that $\tilde{f}_{n,m}(d) $ (indirect approximation) is much more precise than $\tilde{f}_{n,m}(z) $ (direct approximation). For the remainder of this section, we will use the indirect approximation to approximate the Euclidean distance between the transmitter and receiver.
\begin{figure}
    \centering
    {\includegraphics[width=0.96\textwidth]{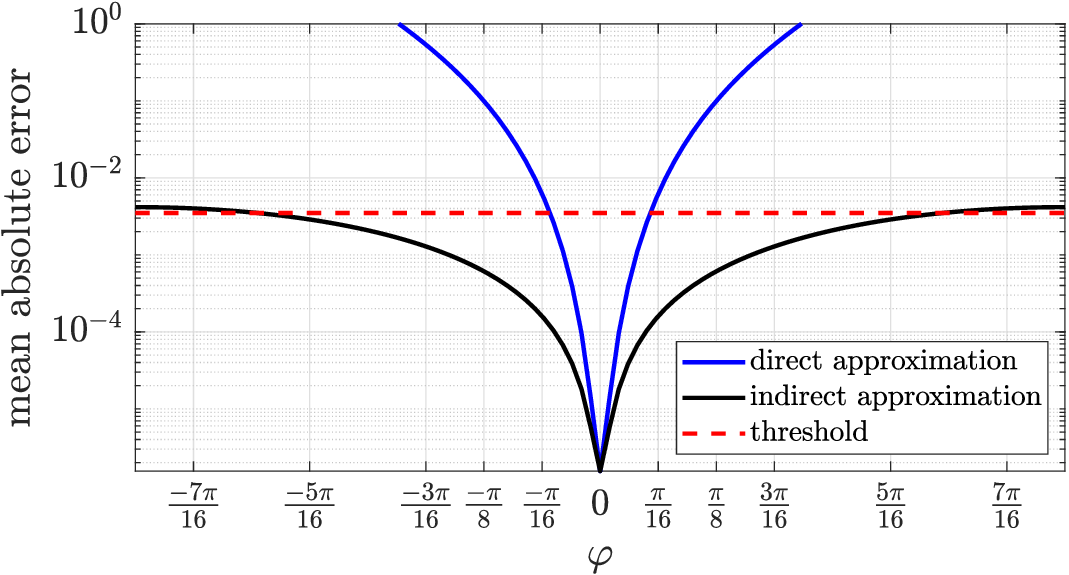}}
    \caption{The evaluation of the Euclidean distance approximations for various azimuth angles $\varphi$. The mean absolute error is obtained by computing the average absolute difference between the exact Euclidean distance and its approximation. We set $f_c = 3$ GHz, $N=10^4$, $ {\eta}=1$, $D=\lambda/4$, and $d=2500 d_F$. The threshold is set to be $3.5 \cdot 10^{-3}$.  }
    \label{Fig_Fresnel_approx_theta}
\end{figure}

\begin{figure} 
    \centering
    {\includegraphics[width=\textwidth]{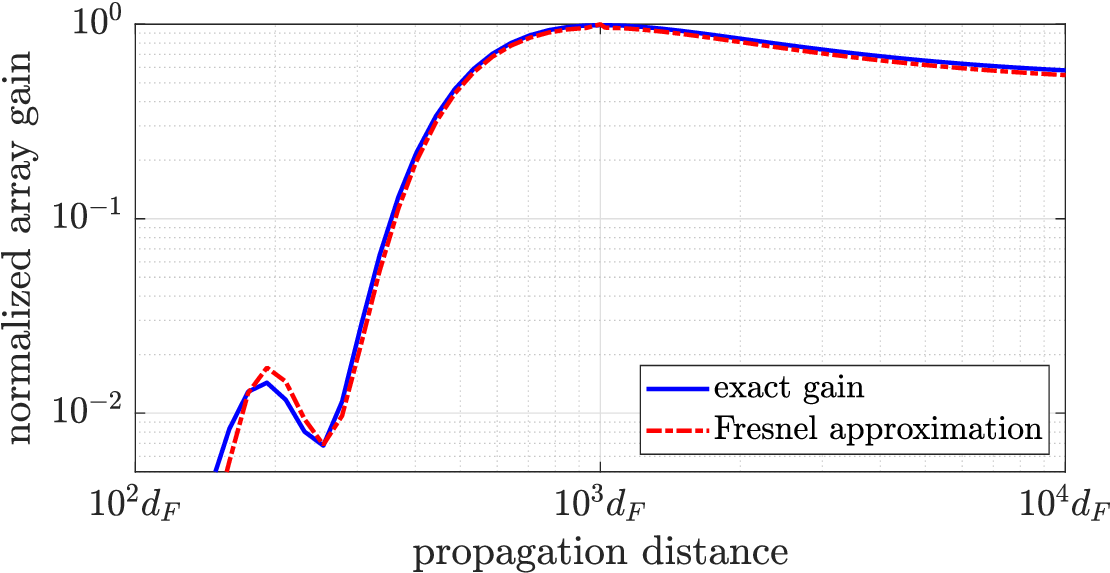}}
    \caption{The Fresnel approximation of the normalized array gain for a non-broadside transmitter with $\varphi=  \pi/16$ and focus on $(0,0,1000d_F)$. We set $ {\eta}=1, N=10^4,$ and $ D=\lambda/4$.} \label{Fig_Fresnel_vs_Exact_theta}
\end{figure}
In the following, we will use the approximation  $\tilde{f}_{n,m}(d) $  to derive an analytical normalized array gain approximation. The Fresnel approximation for the rectangular array with a non-broadside transmitter is written as 
\begin{equation}\label{eq_IV_FresnelApprox}
    E(x,y) \approx \frac{E_0}{\sqrt{4 \pi} z} e^{-j \frac{2 \pi}{\lambda} \left( d + \frac{x^2 + y^2 -2xd \sin(\varphi)}{2d} \right)}.
\end{equation}
We assume that the matched filtering of the array is directed towards the point $(0, 0, F)$, which could potentially differ from the actual transmitter location  $(x_t,0,z)$.  We utilize the Fresnel approximation in \eqref{eq_IV_FresnelApprox} to calculate the normalized array gain by injecting the phase-shift $e^{+j \frac{2\pi}{\lambda} \left(\frac{x^2}{2F} + \frac{y^2}{2F}\right)}$ into the integrals in \eqref{eq_II_NormalizedArrayGain}.
\begin{Theorem}\label{Fresnel_non_Broadside}
When the transmitter is located at $(x_t,0,z)$ and the matched filtering is focused on $(0,0,F)$, the Fresnel approximation of the normalized array gain for the generalized
rectangular array becomes
\begin{equation}\label{eq_IV_ApproxGain_non_Broadside}
\begin{split}
 &\tilde{G}_{\rm rect}( {\eta},\varphi) = \frac{ \left(( {{C}^{2}}( p )+{{S}^{2}}( p ) \right) } {( 2  {\eta} {{p}^{2}} )^{2}} \times \\
    &\left( \left( C(  {\eta} p+q)+C(  {\eta} p-q) \right) ^{2}  + \left( S( {\eta} p+q)+S(  {\eta} p-q) \right)^{2}   \right),
\end{split}
\end{equation}
 where $p=\frac{1}{2} \sqrt{\frac{{{d}_{FA}}}{{{d}_{\rm eff}}(1+{{ {\eta}}^{2}})}}$, $q=\frac{{{x}_{t}}}{d \pi } \sqrt{2\lambda {{d}_{\rm eff}}}=\frac{{\sin (\varphi)}}{\pi } \sqrt{2\lambda {{d}_{\rm eff}}}$, and $d_{\rm eff} = \frac{F d}{|F-d|}$.
\end{Theorem}
\begin{proof}
The proof is provided in Appendix \ref{App_Fresnel_Approx_Non_broadside}.
\end{proof}

Fig.~\ref{Fig_Fresnel_vs_Exact_theta} presents an evaluation of the accuracy of the analytical expression stated in Theorem~\ref{Fresnel_non_Broadside}.  It is observed that the Fresnel approximation in Theorem~\ref{Fresnel_non_Broadside} closely approximates the normalized array gain.  The normalized antenna gain is computed numerically using \eqref{eq_II_NormalizedArrayGain}, $E(x,y)$ as in \eqref{eq_II_ElectField}, and an injected phase-shift. Theorem~\ref{Fresnel_non_Broadside} provides a general statement about arrays with arbitrary transmitter orientations. This is in contrast to Theorem~\ref{Fresnel_Approx_Rect}, in which the transmitter orientation is restricted to the broadside direction. 

Theorem~\ref{Fresnel_Approx_Rect} can be derived as a special case of Theorem ~\ref{Fresnel_non_Broadside}, proved as follows.
\begin{Corollary}\label{cor_theta_0}
The normalized array gain approximation in \eqref{eq_IV_ApproxGain_non_Broadside} reduces to the one in \eqref{eq_III_ApproxGainRect} when $\varphi=0$, indicating a transmitter with a broadside direction to the receiver.
\end{Corollary}
\begin{proof}
 When $\varphi = 0$, it follows that  $\sin(\varphi) =0$  and $\cos(\varphi) =1$. Thus, $x_t =0$ while $z=d$. Then, $p=\sqrt{\frac{d_{FA}}{4d_{\rm eff}(1+ {\eta}^2)}}$ is exactly the same expression as $ a $. Moreover, since  $x_t=0$, $q=0$. Substituting  $p =  a $ and $q=0$ into \eqref{eq_IV_ApproxGain_non_Broadside} yields \eqref{eq_III_ApproxGainRect} which completes the proof.
\end{proof}

The array gain is irrespective of the azimuth angle $\varphi$ when $ {\eta} \rightarrow 0$ (a vertical ULA). Therefore, the BD is independent of $\varphi$. We prove this statement as follows.

\begin{Corollary}\label{cor_c_0}
For $ {\eta} \rightarrow 0$, the normalized array gain approximation is independent of the azimuth angle $\varphi$. Thus, the BD is also independent of $\varphi$.
\end{Corollary}
\begin{proof}
For $ {\eta} \rightarrow 0$, we mathematically express $\tilde{G}_{\rm rect}( {\eta},\varphi)$ as in \eqref{limit_cor_c_0}. We obtain the result of the first limit as $C^2\left(\sqrt{\frac{d_{FA}}{4 d_{\rm eff}}}\right) + S^2\left(\sqrt{\frac{d_{FA}}{4 d_{\rm eff}}}\right)$, which is independent of $\varphi$. To compute the second limit, we first consider $\lim_{ {\eta}\rightarrow 0} C( {\eta} p+q) - \lim_{ {\eta}\rightarrow 0} C(-  {\eta} p+q)$ = $\Delta$, where $\Delta$ is a very small value approaching zero. The same idea applies to $\lim_{ {\eta}\rightarrow 0} S( {\eta} p+q) - \lim_{ {\eta}\rightarrow 0} S(-  {\eta} p+q)=\Delta$. 
 Due to the odd function property of the Fresnel integrals, we can rewrite the second limit in \eqref{limit_cor_c_0} as 
\begin{figure*}
\begin{align}\label{limit_cor_c_0}
\notag
& \lim_{ {\eta}\rightarrow 0} \frac{ \left( {{C}^{2}}\left( p \right)+{{S}^{2}}\left( p \right) \right)}{{\left( 2  {\eta} {{p}^{2}} \right)}^{2}} \big({{\left( C\left(  {\eta} p+q\right)+C\left(  {\eta} p-q\right) \right)}^{2}}  + {\left( S\left(  {\eta} p+q\right)+S\left(  {\eta} p-q\right) \right)}^{2}\big). \\
&=\lim_{ {\eta}\rightarrow 0} \left( {{C}^{2}}\left( p \right)+{{S}^{2}}\left( p \right) \right) \lim_{ {\eta}\rightarrow 0}
    \frac{\big({{\left( C\left(  {\eta} p+q\right)+C\left(  {\eta} p-q\right) \right)}^{2}}  + {\left( S\left(  {\eta} p+q\right)+S\left(  {\eta} p-q\right) \right)}^{2}\big)}{{\left( 2  {\eta} {{p}^{2}} \right)}^{2}}.
\end{align}
\begin{align}\label{limit_cor_c_0_3}
& \lim_{ {\eta}\rightarrow 0}
\notag
    \frac{{{\left( C\left(  {\eta} p+q\right)-C\left( -  {\eta} p+q\right) \right)}^{2}}  + {\left( S\left(  {\eta} p+q\right)-S\left( -  {\eta} p+q\right) \right)}^{2}}{{\left( 2  {\eta} {{p}^{2}} \right)}^{2}}
    \\ \notag
&= 
\frac{{{\left(\lim_{ {\eta}\rightarrow 0} C\left(  {\eta} p+q\right)- \lim_{ {\eta}\rightarrow 0} C\left( -  {\eta} p+q\right) \right)}^{2}} + {\left( \lim_{ {\eta}\rightarrow 0} S\left(  {\eta} p+q\right)- \lim_{ {\eta}\rightarrow 0} S\left( -  {\eta} p+q\right) \right)}^{2}}{{\left(\frac{d_{FA} }{2 d_{\rm eff} } 
 \lim_{ {\eta}\rightarrow 0}     \frac{ {\eta}}{(1+ {\eta}^2)} \right)}^{2}}    \\
&\approx \frac{\Delta^2 +  \Delta^2}{\Delta^2 } \left(\frac{2d_{\rm eff}}{ d_{FA}}\right)^2 = 8\left(\frac{d_{\rm eff}}{d_{FA}}\right)^2.
\end{align}
\hrulefill
\end{figure*}
   \eqref{limit_cor_c_0_3} which is also independent of $\varphi$. Therefore, \eqref{limit_cor_c_0} is independent of $\varphi$ which completes the proof.
\end{proof}

The reason for this result is that a vertical ULA with isotropic antennas has the same spatial resolution in all directions, as can also be observed from rotational invariance. A similar result can be proved for a horizontal ULA, which provides the same BD for any elevation angle if the azimuth angle is $\varphi=0$.

In what follows, we conduct a numerical analysis to compute the $3$ dB BD of the rectangular array in the presence of a non-broadside transmitter with respect to its width/length proportion $ {\eta}$. We consider both fixed area and length scenarios.

\subsection{Fixed Array Aperture Area vs. Fixed Array Aperture Length}

As mentioned in Section~\ref{S_Rect_length_vs_area}, the array's aperture area can be kept constant at $A_{\rm array} = N\lambda^2/32$ while adjusting the shape of antenna elements through $ {\eta}$. 
The aperture length then varies with $ {\eta}$ as $D_{\rm array} = A_{\rm array}(1+ {\eta}^2)/ {\eta}$. 
Alternatively, fixing $D_{\rm array} = \sqrt{N}\lambda/4$ makes the aperture area $A_{\rm array} = {\eta}(D_{\rm array})^2/(1+ {\eta}^2)$ dependent on $ {\eta}$. Both cases are analyzed below. 

\begin{figure}
\centering
\subfloat[The BD of rectangular arrays with a fixed area versus the azimuth angle.]
{\includegraphics[width=\textwidth]{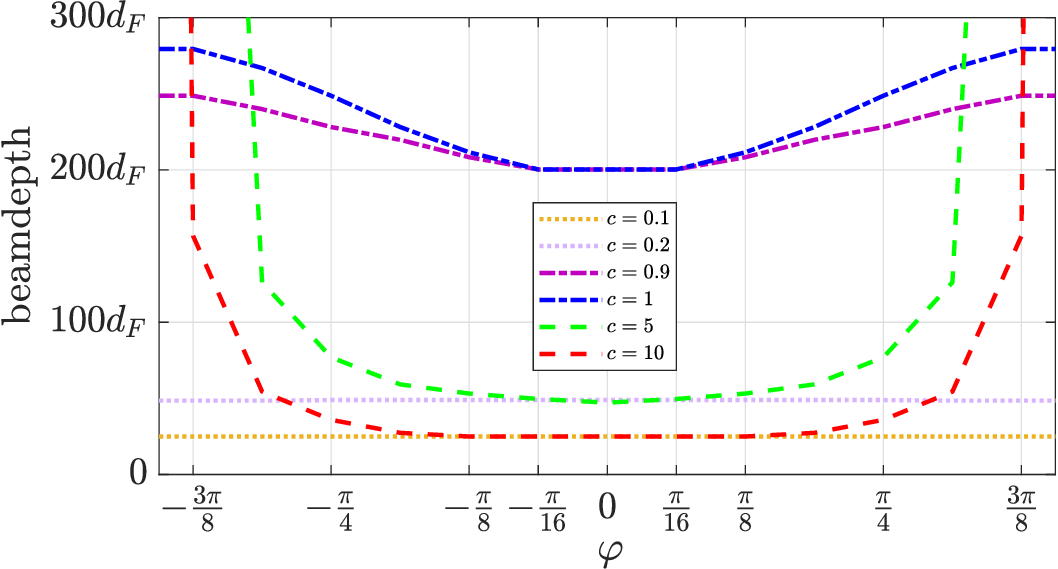}}\hfill
\centering
\subfloat[The BD of rectangular arrays with a fixed length versus the azimuth angle.]
{\includegraphics[width=\textwidth]{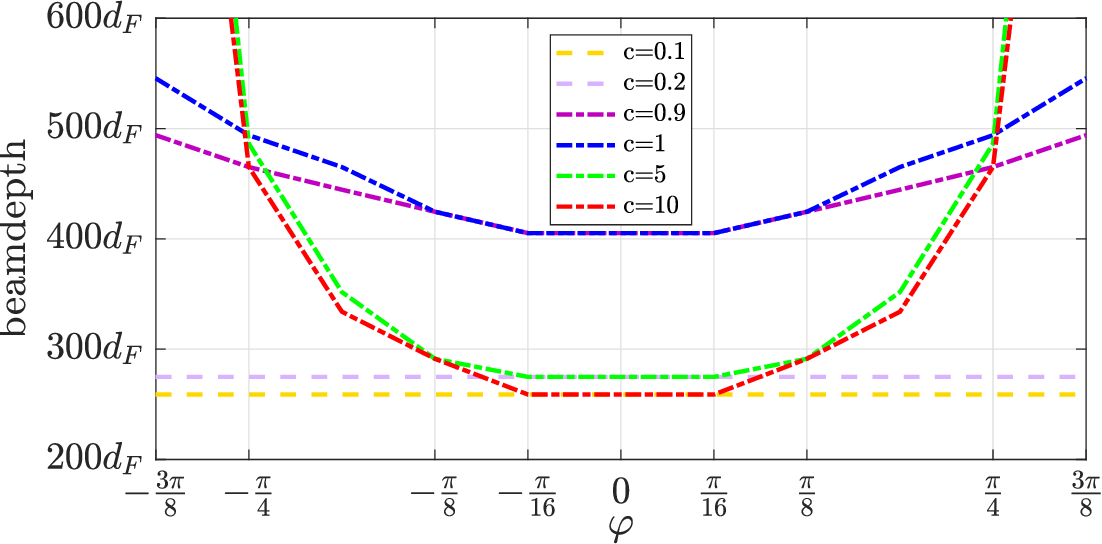}}\hfill
\centering
\caption{The BD for rectangular array shapes that are either tall ($ {\eta}<1$), square ($ {\eta}=1$), or wide ($ {\eta}>1$)  arrays. The BD is shown with respect to the azimuth angle for  $f_c = 3$ GHz and $N=10^4$. }
\label{Fig_Beamdepth_vs_theta}
\end{figure}

\begin{figure*}
\centering
\subfloat[$ {\eta}=0.1$]
{\includegraphics[scale=0.45]{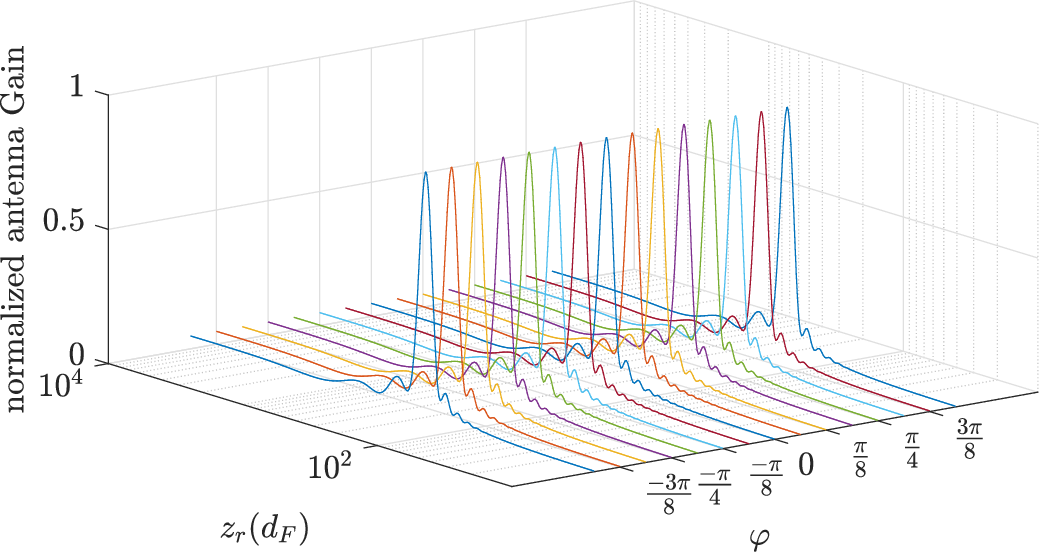}}\hfill
\centering
\subfloat[$ {\eta}=10$]
{\includegraphics[scale=0.45]{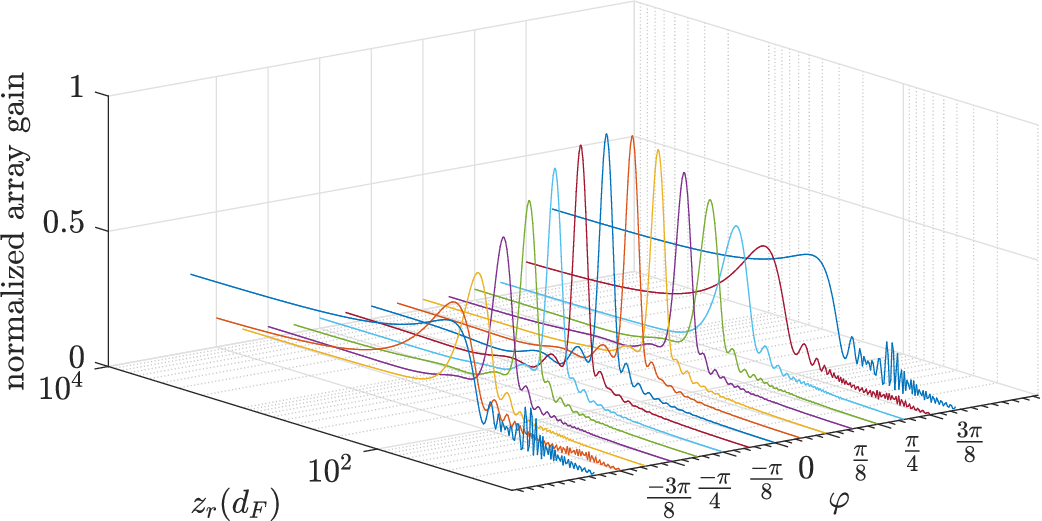}}
\caption{The normalized array gain for different azimuth angles $\varphi$ and propagation distance $d$. We set $f_c = 3$ GHz, $N=10^4$, and a fixed area of $N \lambda^2/32$.}
\label{Fig_array_gains}
\end{figure*}

In Fig.~\ref{Fig_Beamdepth_vs_theta}a, we plot the $3$ dB BD for various rectangular array shapes with a fixed area and a transmitter located in some angular direction in the range $-3 \pi/8 \leq \varphi \leq 3 \pi/8$. Let us consider a BD variation, defined as the difference between the largest and smallest BD values with respect to $\varphi$, for a particular value of $ {\eta}$. The BD variation becomes greater as $ {\eta}$ increases. When $ {\eta} \rightarrow 0$, the BD variation approaches zero, implying that the BD is unaffected by the azimuth angle $\varphi$. This finding agrees with Corollary~\ref{cor_c_0}. For a square array ($ {\eta}=1$), the BD is around $550 d_F$ for  $\varphi= \pm 3\pi/8$. In the case of higher $ {\eta}$ (e.g., $ {\eta}=10$), the BD becomes much larger than $600 d_F$. The latter indicates that the BD approaches infinity when $\varphi \rightarrow \pm \pi/2$. 
To illustrate this better, we plot Fig.~\ref{Fig_array_gains}, which shows the normalized antenna gains for a particular value of $ {\eta}$ with respect to propagation distances and azimuth angles. Fig.~\ref{Fig_array_gains}b shows a larger variation in normalized array gain compared to Figs.~\ref{Fig_array_gains}a. Furthermore, the BD for $\varphi= \pm {14 \pi}/{32}$ in Fig.~\ref{Fig_array_gains}b is infinite when $d \geq dB$, where the normalized array gain is around $0.3$, similar to its peak.

Let us now consider the scenario where the array's aperture length is fixed, but it depends on the value of $ {\eta}$. The $3$ dB BD for various rectangular array shapes is shown in Fig.~\ref{Fig_Beamdepth_vs_theta}b. Similar to the previous scenario with a fixed aperture area, we observe that the variation in BD increases as $ {\eta}$ increases, and the BD becomes independent of the azimuth angle $\varphi$ as $ {\eta}$ approaches zero. Furthermore, we observe that the highest BD values for arrays with $ {\eta}=0.1$ and $ {\eta}=1$ differ by approximately a factor of two, as the former has a value of around $550 d_F$ while the latter has a value of around $260 d_F$. This is in contrast to the scenario with a fixed array area, where the highest BD values for arrays with $ {\eta}=0.1$ and $ {\eta}=1$ differ by approximately a factor of eleven as the BD value of an array with $ {\eta}=0.1$ is around $560 d_F $ while the one with $ {\eta}=1$ is around $50 d_F $. This  indicates that the aperture length is a dominant factor determining the BD.  
\begin{Observation}
As $ {\eta}$ approaches zero, so that the array approaches a vertical ULA, the BD becomes independent of the azimuth angles. For larger values of $ {\eta}$, the BD increases with increasing azimuth angle $\varphi$, and approaches infinity as $\varphi$ approaches $\pm \frac{\pi}{2}$. Finally, as the value of $ {\eta}$ increases, the  BD variation also increases.
\end{Observation}

\subsection{Projected Rectangular Array}

In this section, we investigate the feasibility of approximating the gain that a rectangular array achieves with non-broadside transmission using a smaller array that is rotated for broadside transmission.
Different from the precise analysis in the preceding subsection, this approximation will give simple qualitative insights.
The smaller array is obtained by projecting the original array at an angle of $\varphi$. The projected array is perpendicular to the transmitter's axis, as shown in Fig.~\ref{Fig_non_broadside_projection}. Notably, the width of the projected array appears shorter by a factor of $\cos(\varphi)$, while the array's height remains unchanged. 
\begin{figure}
    \centering
    {\includegraphics[width=0.55\textwidth]{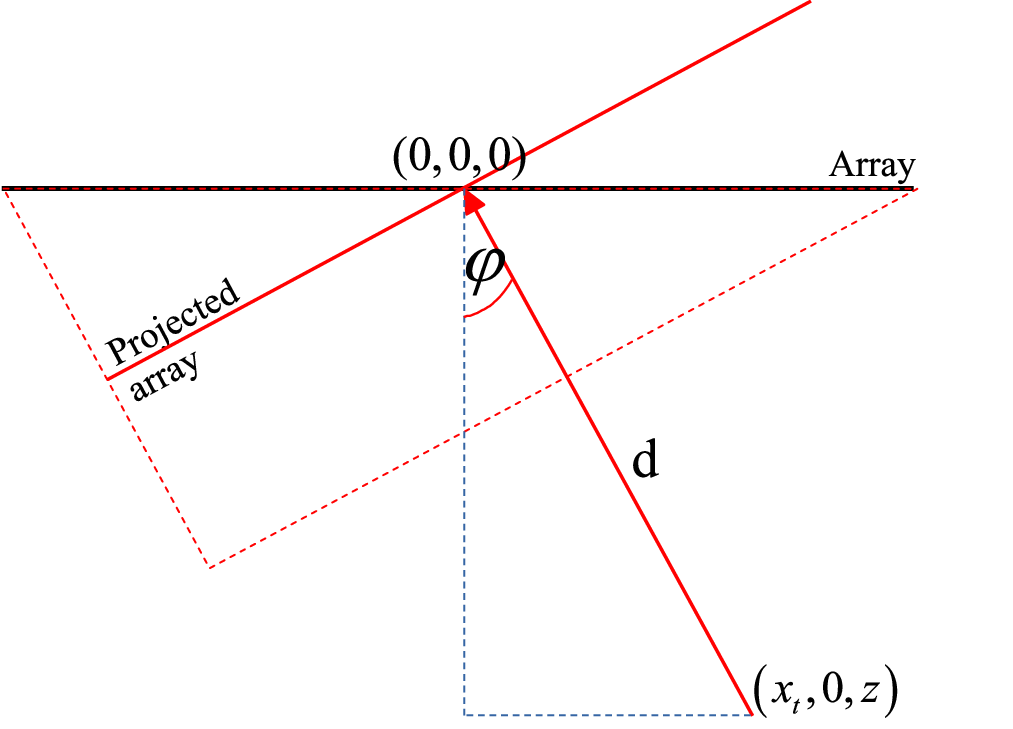}}
    \caption{The broadside transmission to a projected array. The transmitter is located at the distance $d$ with an azimuth angle of $\varphi$ from the center of the receiver array.}
    \label{Fig_non_broadside_projection}
\end{figure}
To evaluate the tightness of the approximation, we plot the normalized array gain versus the propagation distance in Fig.~\ref{Fig_rorated_transmitter_simulations}a. We can see that the approximation is very close to the exact gain of an array with a non-broadside transmitter when the propagation distance $d> 200 d_F$, while they are slightly different for  $100 d_F\leq d< 200 d_F$. Although there exists a minor discrepancy between the exact gain and the approximation, the two values are still closely matched.
We further evaluate the approximation with respect to the azimuth angle $\varphi$ by plotting the absolute difference of the exact and approximation gain in Fig.~\ref{Fig_rorated_transmitter_simulations}b. The approximation increases with $\varphi$. Nevertheless, the error is below $0.1$, which is relatively small. This indicates that the approximation is accurate, thus, we can view an array with a non-broadside transmitter as a projected array with a  broadside transmitter. 

 {Since we can use the projected array approximation for non-broadside transmissions, we can utilize the derived closed-form finite depth limit expression for broadside transmission to evaluate the communication performance with respect to the azimuth angle $\varphi$. More specifically, we find focus points that are associated with non-overlapping $3$ dB BD (as what we have discussed in Section III) based on the $3$ dB BD range $\left[\frac{d_{FA} F }{d_{FA} + 4 F a_{3{\rm dB}} (1+ {\eta}^2) } , \frac{d_{FA} F }{d_{FA} - 4 F a_{3{\rm dB}} (1+ {\eta}^2) } \right] $ for each non-broadside transmitter with angle of $\varphi$. Utilizing those non-overlapping $3$ dB BD for each value of $\eta$, we evaluate the achievable sum rate and plot them in Fig. \ref{fig_sum_rate_phi}.}

Another notable observation is that as the azimuth angle $\varphi$ approaches $\pi/2$ (end-fire direction),  the aperture length of the projected array, denoted as $D_{\rm array}^{\rm proj}$, approaches the height of the original array, which is $\ell$. Mathematically, we can express it as  $\lim_{\varphi \rightarrow \pi/2} (\sqrt{\ell^2 + (w\cdot\cos(\varphi))^2} = \ell $. If $ {\eta} \rightarrow \infty$ (corresponding to a horizontal ULA), then $\ell \rightarrow 0$. According to Theorem~\ref{BD_analytical_rectangular} in this scenario, the $3$ dB BD will tend towards infinity because the denominator shrinks much faster than the numerator, due to the dominant factor $d_{FA}^2 = \left(\frac{2 (D_{\rm array}^{\rm proj})^2}{\lambda}\right)^2$. This observation is consistent with the findings illustrated in Fig.~\ref{Fig_Beamdepth_vs_theta}, which shows that the BD increases when $\varphi \rightarrow \pi/2$, and tends towards infinity when $\varphi $ approaches $ \pi/2$ and $ {\eta}$ is much larger than $1$.
We summarize the findings as follows.

\begin{Observation}\label{Projected_array_non_broadside}
We can approximate the gain of a rectangular array with a non-broadside transmitter using a smaller/projected array with a broadside transmitter. Moreover, as the transmitter approaches the end-fire direction of the array, the BD grows and approaches infinity as $ {\eta}$ becomes much greater than $1$.
\end{Observation}

\begin{figure}
\centering
\subfloat[The normalized array gain versus the propagation distance, $\varphi=\pi/4$.]
{\includegraphics[width=1.02\textwidth]{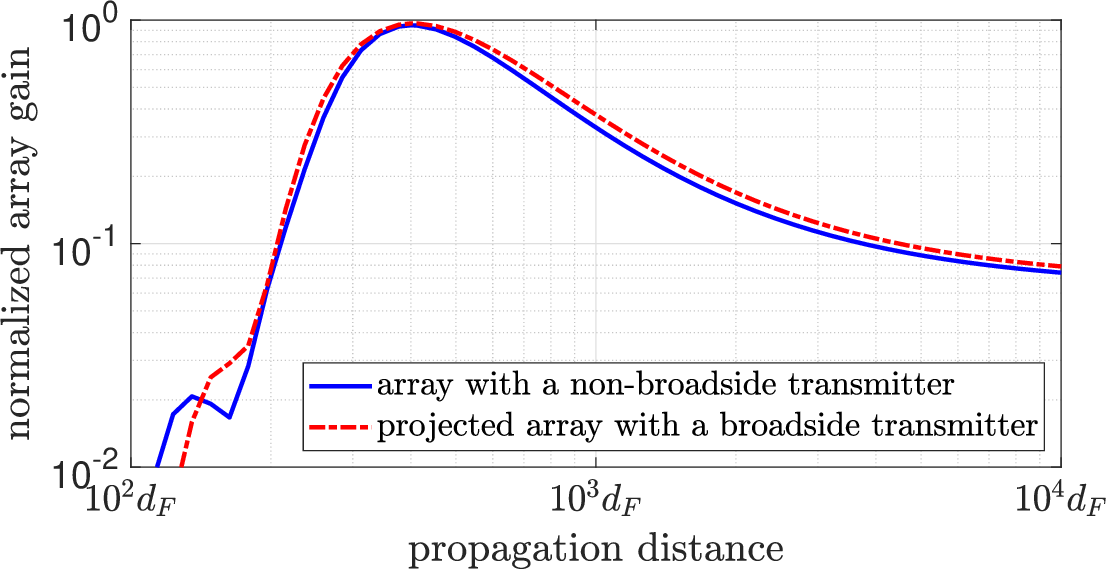}}\hfill
\centering
\subfloat[The approximation error versus the azimuth angle, $d=1000d_F$.]
{\includegraphics[width=0.99\textwidth]{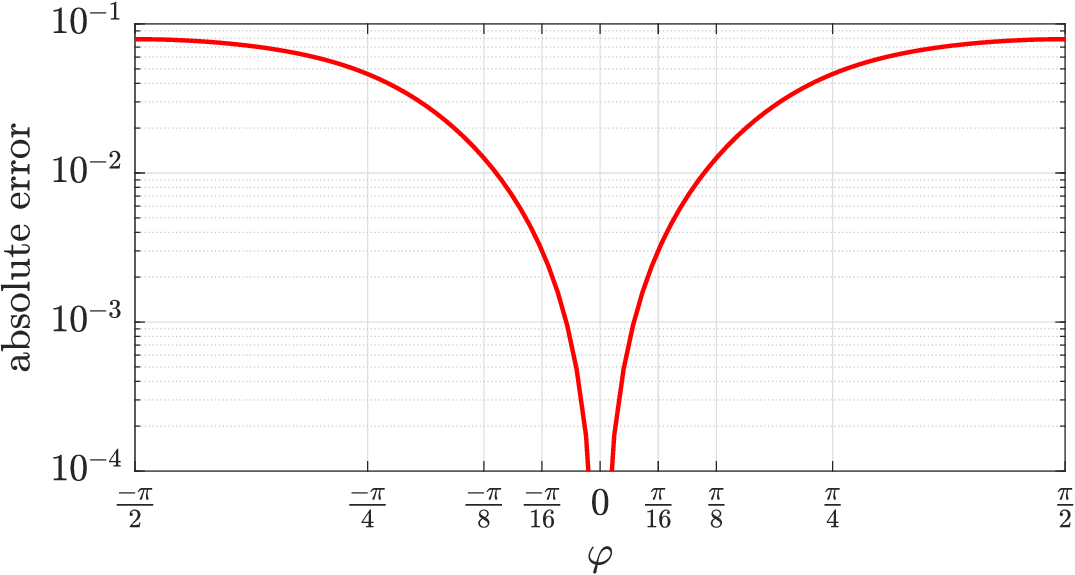}}
\caption{The normalized array gain achieved by a square array with a non-broadside transmitter can be approximated by a projected rectangular array with a broadside transmitter. We set $F=d_B=400d_F$, $f_c = 3$ GHz, and $N=10^{4}$. }
\label{Fig_rorated_transmitter_simulations}
\end{figure}

\begin{figure}
    \centering
    \includegraphics[width=0.97\textwidth]{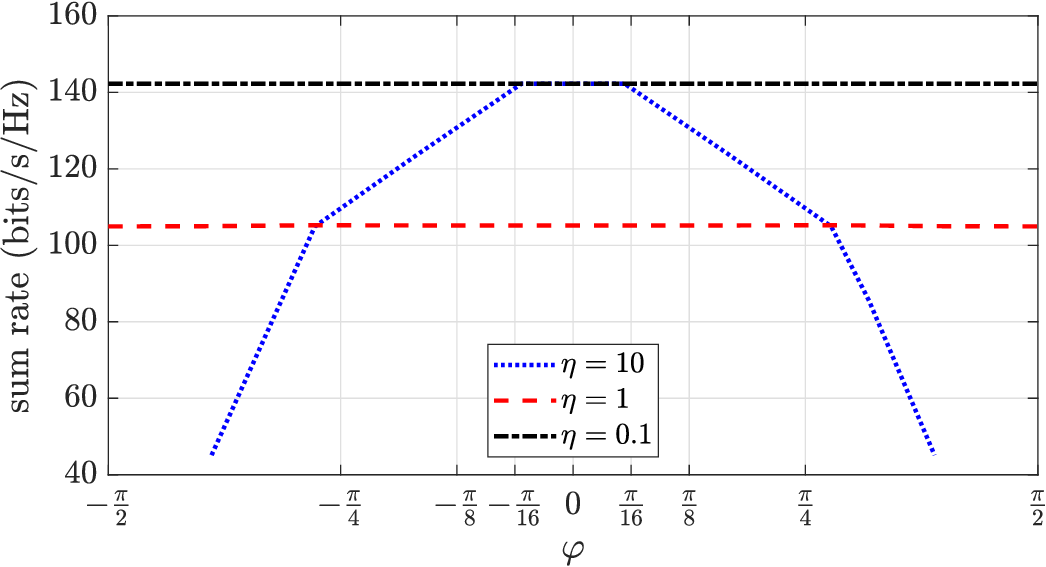}
    \caption{ {The achievable sum rate with respect to the azimuth angle $\varphi$ with non-overlapping $3$dB BD users.}}
    \label{fig_sum_rate_phi}
\end{figure}

\section{Gain and Beam Depth of a Circular Array}
\label{Sect_Beamdepth_Circ}

 In this section, we consider a UPA with a circular shape. The circular array may provide constant consistent effective array aperture at all angles, thanks to its rotational symmetry property \cite{2022_Wu_Arxiv}. To implement the circular array in practice, one can use the concept of concentric circular arrays \cite{2007_Chan_TSP,2005_Ioannides_WPL},  where the circular array is divided into multiple concentric rings. Every ring is uniformly divided in the angular domain into segments, each corresponding to an antenna element. The number of antenna elements is determined to achieve sufficiently small element areas. If this is done to prevent significant gain variations across the elements, then the theoretical results provided in this section are achievable.  Since we consider sufficiently small antenna elements, the array's overall gain remains minimally affected by a specific element configuration. 
  
  This presumption relies on the idea that the antenna elements in the array will radiate in a roughly symmetrical way. Consequently, the results are independent of the number of antennas, provided it is adequate to accommodate suitably small antenna elements across the array. The Fresnel approximation of the normalized array gain for a circular array with a transmitter located at $(0,0,z)$ and the focus point at $(0,0,F)$, is computed using \eqref{eq_II_NormalizedArrayGain} with $E(x,y)$ in \eqref{eq_III_FresnelApprox} and inject the phase-shift $e^{+j \frac{2\pi}{\lambda} \left(\frac{x^2}{2F} + \frac{y^2}{2F}\right)}$. Due to the circular geometry of the array, we adopt a polar coordinate system in the following analysis.

\subsection{Gain and $3$ dB Beam Depth}

We will now analyze the normalized array gain and BD of the circular array. Assuming that the phase shifts are perfectly compensated by continuous matched filtering. When we inject a phase shift to focus the array on a certain propagation distance, the normalized array gain can be approximated using the Fresnel approximation, given in the following theorem. 
\begin{Theorem}\label{Fresnel_Approx_Circ}
When the transmitter is located at $(0, 0, z)$ and the matched filtering is focused on $(0, 0, F)$, the normalized array gain for the circular array can be approximated using the Fresnel approximation as
\begin{equation}\label{eq_V_ApproxGainCirc}
    {{\hat{G}}_{\text{circ}}}={ \operatorname{sinc}^2 ( {l})},
\end{equation}
 where $ {l} = \frac{R^2}{2\lambda z_{\rm eff}}$ and $R = D/2 = \lambda/8$
 is the radius of the circular array, which matches the aperture length of the rectangular array. 
\end{Theorem}
\begin{proof}
The proof is provided in Appendix \ref{App_Fresnel_Approx_Circular} and is inspired by \cite{Silver_BOOK}.
\end{proof}

The $3$ dB BD for the circular array when focusing on a point $F\geq d_B$ can then be calculated based on \eqref{eq_V_ApproxGainCirc}.
\begin{Theorem}\label{BD_analytical_circular}
The $3$ dB BD of a circular array with a matched filter focused on $(0,0,F)$, where $F \geq d_B $, is computed as
\begin{equation}\label{eq_V_BD_Circular}
 {\rm BD}_{3 {\rm dB}}^{\rm circ}  = 
        \begin{cases}
        \frac{1.7720 R^2 F^2 \lambda}{R^4- \left(0.886 \lambda F\right)^2}  , & F<\frac{R^2}{0.886 \lambda},\\
        \infty, & F \geq \frac{R^2}{0.886 \lambda}.
        \end{cases}
\end{equation}
\end{Theorem}
\begin{proof}
We know that ${\sf sinc}^2(0)=1$ and ${\sf sinc}^2(0.4430) \approx 0.5$. The $3$ dB BD ${\rm BD}_{3 {\rm dB}}^{\rm circ}$ is defined as the range of distances where the gain is greater  than $0.5$, therefore, it is the difference of two possible values of $z$, i.e., $\frac{R^2 F}{R^2 - 0.886 \lambda F} -  \frac{R^2 F}{R^2 + 0.886 \lambda F} $ which results in \eqref{eq_V_BD_Circular}. 
\end{proof}

The finite BD limit for the circular array is $\frac{R^2}{0.886 \lambda}$, since  the BD goes to infinity when the focus is on a point larger than the finite BD limit.  
We evaluate the normalized array gain of the circular array in Fig.~\ref{Fig_Gain_circular}. We can see that we approach $-3$ dB as $z \rightarrow \infty$, and hence the finite BD limit $\frac{R^2}{0.886 \lambda}$ is the largest distance in which the BD for which the BD is finite. 

If we  adjust the focus to $F=d_B=2(2R)$ and keep the array's aperture length the same as in the case of a rectangular array (i.e., $R=D/2$), then we obtain the BD for the circular array as ${\rm BD}_{3 {\rm dB}}^{\rm circ} \approx 247 d_F$, according to Theorem~\ref{BD_analytical_circular}. The BD is slightly higher as compared to the ${\rm BD}_{3 {\rm dB}}^{\rm rect} \approx 244 d_F$ in Theorem~\ref{BD_analytical_rectangular} for tall  $ {\eta}=0.1$ or wide $ {\eta}=10$ rectangular arrays.

\subsection{Nulls and Side-lobes in the Distance Domain}

\begin{figure}
    \centering
    {\includegraphics[width=1.03\textwidth]{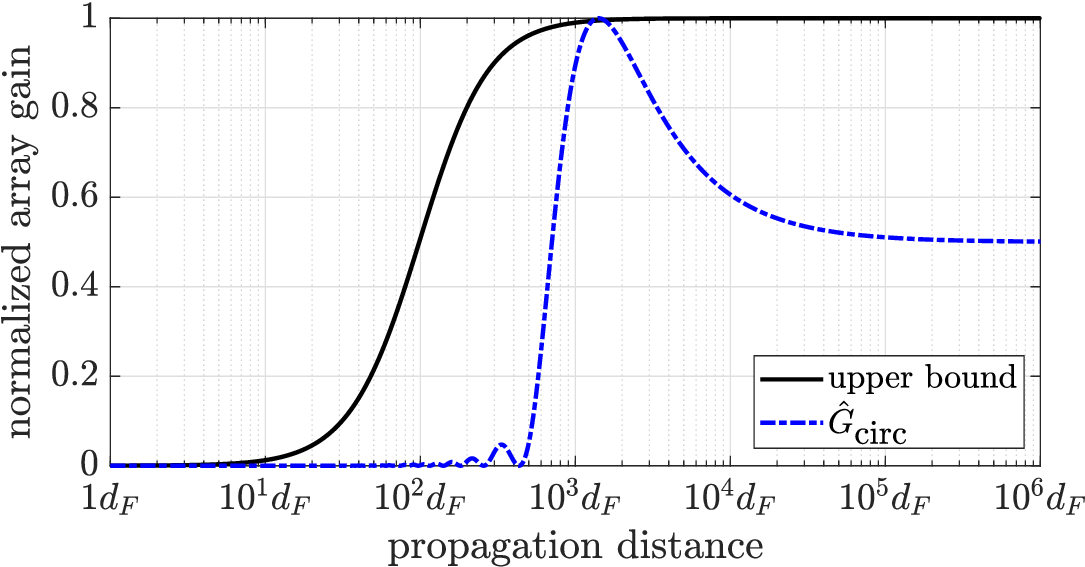}}
    \caption{The normalized circular array gain versus the propagation distance.  }
    \label{Fig_Gain_circular}
\end{figure}

\begin{figure}
\centering
{\includegraphics[width=\textwidth]{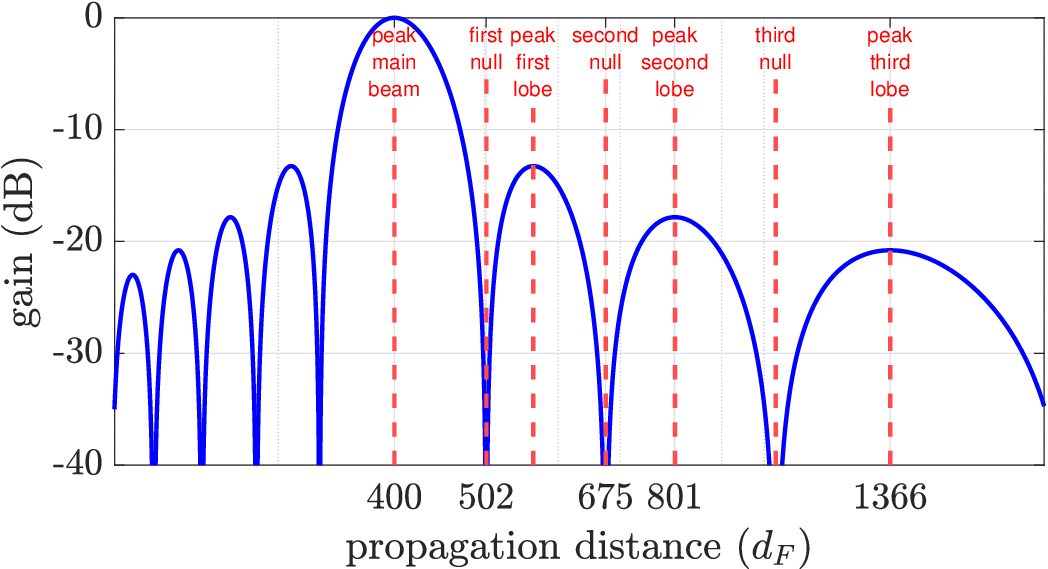}}
\caption{The array gains, nulls, and side-lobes with respect to the propagation distance with a circular array, $F =d_B$.}
\label{Fig_Nulls_Sidelobes_circular}
\end{figure}

\begin{table}[ht!]\small
\begin{center}
\caption{Summary of nulls and side-lobes in the distance domain with a circular array.}
\label{tab_nulls_sidelobes}
\begin{tabular}{|c|c|c|}
\hline
      \textbf{Sinc constant} $ {l} = \frac{R^2}{2\lambda z_{\rm eff}}$ & \textbf{Gain (dB)} & \textbf{Explanation}\\
      \hline\hline
      0 & 0 & peak of main beam\\\hline
      1 & $-\infty$ & first null\\\hline
      1.43 & $-13.26$ & peak of first lobe\\\hline
      2 & $-\infty$ & second null\\\hline
      2.46 & $-17.83$ & peak of second lobe\\\hline
      3 & $-\infty$ & third null\\\hline
      3.47 & $-20.79$ & peak of third lobe\\\hline
\end{tabular} \vspace{-5mm}
\end{center}
\end{table}

When beams are analyzed in the angular domain, they are known to consist of a main beam and multiple weaker side-lobes, with nulls in between. The same behavior appears in the distance domain when considering finite-depth beamforming. We will now investigate the nulls and side-lobes for the circular array based on the derived gain in \eqref{eq_V_ApproxGainCirc}.  {The derived array gain conforms to the sinc function, which allows us to characterize its initial peak, nulls, and other peaks.}
The peak of the main beam in the distance domain appears when $ {l}=0$ so that $\left(\frac{{\sin}( {l})}{ {l}}\right)^2=1$. Since  $ {l} = \frac{\pi R^2}{2\lambda z_{\rm eff}}$, it equals to $0$ when $z_{\rm eff}=\infty$ which implies that $F=z$. Therefore, the peak of the main beam is obtained at the distance $F=z$. The first null appears when $ {l}=\pi$. Hence, $z_{\rm eff} = \frac{R^2}{2 \lambda}$. The second and later nulls appear when $ {l}=2\pi, \ldots, K\pi$. Therefore, the nulls can be obtained by setting $z_{\rm eff} = \frac{k_0 R^2}{2 \lambda}$, where $k_0 \in \{1,\ldots,K\}$ indicates the index of the nulls. The peak values of the side-lobes is given in Table~\ref{tab_nulls_sidelobes}. We plot Fig.~\ref{Fig_Nulls_Sidelobes_circular} to demonstrate the nulls and side-lobes of the circular array. The characterization of the nulls and beamdepths in Table~\ref{tab_nulls_sidelobes} match the nulls and beamdepths in Fig.~\ref{Fig_Nulls_Sidelobes_circular}. 

\section{Conclusion}
\label{Sect_conclusion}

In this paper, we analyzed how the BD of a large planar rectangular array depends on its width-to-height proportion, area, and aperture length. We derived a tight Fresnel approximation of the normalized array gain and used it to characterize the BD of the rectangular array. We proved that the $3$ dB BD is the largest for a square array. The BD pattern is symmetric for tall and wide arrays. The smallest BD is obtained when the array approaches a ULA. Therefore, an array of linear geometry is preferred since it offers the smallest BD, enabling us to better utilize spatial multiplexing since there is a higher degree of spatial orthogonality (in the distance domain). We also considered a projected square array where its effective aperture length is shorter than the original square array due to the rotation. 
 For a non-broadside transmitter, where we first evaluated the Taylor approximation of the Euclidean between the transmitter and the antenna element of the array. We refined the approximation so that the approximation error decreases with the increase of the propagation distance. When the width-to-height proportion approaches zero, the BD becomes independent of the azimuth angles. For larger width-to-height proportions, the BD increases with increasing azimuth angle. In addition, as the value of width-to-height proportion $ {\eta}$ increases, the BD variation also increases. Furthermore, we demonstrated that the beam pattern of an array with a non-broadside transmitter can be approximated by that of a projected/smaller array with a broadside transmitter. The original array was projected onto a plane that was perpendicular to the axis of the non-broadside transmitter. Finally, we considered a UPA with a circular shape. The BD of the circular array is slightly larger compared to that of the wide/tall rectangular array, but lower than the square array. Our analysis also involved characterizing the nulls and side lobes of the circular array that emerge at different distances.  {The fundamental results of this paper can guide antenna array design for near-field communications, including their hardware capabilities.}

\appendices

\section{Proof of Theorem~\ref{Fresnel_Approx_Rect}}
\label{App_Fresnel_Approx}

First, we substitute the Fresnel approximation in \eqref{eq_III_FresnelApprox} into the normalized array gain in \eqref{eq_II_NormalizedArrayGain} and inject phase-shift $e^{+j \frac{2\pi}{\lambda} \left(\frac{x^2}{2F} + \frac{y^2}{2F}\right)}$, which gives us 
\begin{align}
    \notag
    & \hat{G}_{\rm rect}( {\eta}) =\frac{1}{{{(N A)}^{2}}} \times
    \\ \notag
    &\Bigg| \int\limits_{-\frac{D_{\rm array}}{2\sqrt{1+{{ {\eta}}^{2}}}}}^{\frac{D_{\rm array}}{2\sqrt{1+{{ {\eta}}^{2}}}}}{\int\limits_{-\frac{ {\eta} D_{\rm array}}{2\sqrt{1+{{ {\eta}}^{2}}}}}^{\frac{ {\eta} D_{\rm array}}{2\sqrt{1+{{ {\eta}}^{2}}}}} 
   e^{j \frac{2\pi}{\lambda} \left(\frac{x^2}{2F} + \frac{y^2}{2F}\right)}  {e}^{-j\frac{2\pi }{\lambda }\left( z+\frac{{{x}^{2}}}{2z}+\frac{{{y}^{2}}}{2z} \right)}  }dxdy \Bigg|^{2} \\ 
   &    =\frac{1}{{{(N A)}^{2}}}{{\left| \int\limits_{-\frac{D_{\rm array}}{2\sqrt{1+{{ {\eta}}^{2}}}}}^{\frac{D_{\rm array}}{2\sqrt{1+{{ {\eta}}^{2}}}}}{\int\limits_{-\frac{ {\eta} D_{\rm array}}{2\sqrt{1+{{ {\eta}}^{2}}}}}^{\frac{ {\eta} D_{\rm array}}{2\sqrt{1+{{ {\eta}}^{2}}}}} \!\!\! {{{e}^{-j\frac{2\pi }{\lambda } \frac{\left( |F-z| \right)\left( {{x}^{2}}+{{y}^{2}} \right)}{2zF} }}}}dxdy \right|}^{2}}.
\end{align}
By defining $z_{\rm eff} = \frac{Fz}{|F-z|}$, we can rewrite the expression as
\begin{align}\label{G-derivation}
    \hat{G}_{\rm rect}( {\eta})  =\frac{1}{(N A)^{2}}{{\left| \int\limits_{-\frac{D_{\rm array}}{2\sqrt{1+{{ {\eta}}^{2}}}}}^{\frac{D_{\rm array}}{2\sqrt{1+{{ {\eta}}^{2}}}}}\int\limits_{-\frac{ {\eta} D_{\rm array}}{2\sqrt{1+{{ {\eta}}^{2}}}}}^{\frac{ {\eta} D_{\rm array}}{2\sqrt{1+{ {\eta}}^{2}}}}
    {e}^{-j\frac{\pi }{\lambda } \frac{ {{x}^{2}}}{z_{\rm eff}}  }  {e}^{-j\frac{\pi }{\lambda } \frac{ {{y}^{2}}}{z_{\rm eff}}  }dxdy \right|}^{2}}.
\end{align}
The evaluation of the anti-derivatives in \eqref{G-derivation} yields  \cite{1956_Polk_TAP}
\begin{align}
    \notag
    &\hat{G}_{\rm rect}( {\eta})  =\left( \frac{4 z_{\rm eff}(1+{{ {\eta}}^{2}})}{ {\eta}  d_{FA}} \right)^{2} \times
     \\   \notag
    &\Bigg( {C}^{2} \Bigg(  {\eta}\sqrt{\frac{d_{FA}}{4{z}_{\rm eff} \left({1+{{ {\eta}}^{2}}}\right)}}  \Bigg) +{{S}^{2}}\Bigg(  {\eta} \sqrt{\frac{d_{FA}}{4{z}_{\rm eff} \left({1+{{ {\eta}}^{2}}}\right)}} \Bigg) \Bigg) \times
    \\  
    &  {C}^{2} \Bigg( \sqrt{\frac{d_{FA}}{4{z}_{\rm eff} \left({1+{{ {\eta}}^{2}}}\right)}}  \Bigg)+{{S}^{2}}\Bigg( \sqrt{\frac{d_{FA}}{4{z}_{\rm eff} \left({1+{{ {\eta}}^{2}}}\right)}} \Bigg) ,
\end{align}
where $C\left(\cdot \right)$ and $S\left(\cdot \right)$ are the Fresnel integrals.  By defining $ a = \frac{d_{FA}}{4{z}_{\rm eff}(1+{{ {\eta}}^{2}})}$, we finally obtain  
\begin{equation}
\hat{G}_{\rm rect}( {\eta}) =
\frac{\left( {{C}^{2}}\left(  {\eta}\sqrt{ a } \right)+{{S}^{2}}\left(  {\eta}\sqrt{ a } \right) \right) \left( {{C}^{2}}\left( \sqrt{ a } \right)+{{S}^{2}}\left( \sqrt{ a } \right) \right)}{(  {\eta} a )^{2}}.
\end{equation}
This completes the proof.

\vspace*{-1.5cm}

 {
\section{Proof of Theorem~\ref{Fresnel_non_Broadside}}
\label{App_Fresnel_Approx_Non_broadside}
}

\vspace*{-1cm}

 {
For a transmitter located at $(x_t,y_t,z)$, we use the Fresnel approximation based on \eqref{eq_IV_FresnelApp}, as  
\begin{equation}\label{eq_II_FresnelApproxNew}
    E(x,y) \approx \frac{E_0}{\sqrt{4 \pi} z} e^{-j \frac{2 \pi}{\lambda} \left( d + \frac{x^2 + y^2 -2(x x_t + y y_t)}{2d} \right)}.
\end{equation}
Substituting \eqref{eq_II_FresnelApproxNew}  into the normalized array gain in \eqref{eq_II_NormalizedArrayGain} and inject phase-shift $e^{+j \frac{2\pi}{\lambda} \left(\frac{x^2}{2F} + \frac{y^2}{2F}\right)}$ yields
\begin{multline}\label{temp_eq}
    \tilde{G}_{\rm rect}( {\eta},\varphi) =\frac{1}{{{(N A)}^{2}}} \Bigg|  \int\limits_{-\frac{D_{\rm array}}{2\sqrt{1+{{ {\eta}}^{2}}}}}^{\frac{D_{\rm array}}{2\sqrt{1+{{ {\eta}}^{2}}}}}\int\limits_{-\frac{ {\eta}  D_{\rm array}}{2\sqrt{1+{{ {\eta}}^{2}}}}}^{\frac{  {\eta} D_{\rm array}}{2\sqrt{1+{{ {\eta}}^{2}}}}} {e}^{j\frac{2\pi }{\lambda }\left( \frac{{{x}^{2}}}{2F}+\frac{{{y}^{2}}}{2F} \right)} \times \\
    {e}^{-j\frac{2\pi }{\lambda }\left( d + \frac{x^2 + y^2 -2x x_t -2y y_t}{2d} \right)}dxdy \Bigg|^{2}
\end{multline}
By arranging the variables in \eqref{temp_eq}, we can obtain
\begin{align}
    \notag
    \tilde{G}_{\rm rect}( {\eta},\varphi)&=\frac{1}{{{(N A)}^{2}}}\Bigg| {{e}^{j\left(   \frac{\lambda {{d}_{\rm eff}}}{\pi }  \left(\frac{x_t}{d}\right)^2   \right)}}   \int\limits_{-\frac{D_{\rm array}}{2\sqrt{1+{{ {\eta}}^{2}}}}}^{\frac{D_{\rm array}}{2\sqrt{1+{{ {\eta}}^{2}}}}} \int\limits_{-\frac{cD_{\rm array}}{2\sqrt{1+{{ {\eta}}^{2}}}}}^{\frac{cD_{\rm array}}{2\sqrt{1+{{ {\eta}}^{2}}}}}    
    \\
    &{{{e}^{-j{{\left( \sqrt{\frac{\pi }{\lambda {{d}_{\rm eff}}}}x+\frac{{{x}_{t}}}{d} \sqrt{\frac{\lambda {{d}_{\rm eff}}}{\pi }} \right)}^{2}}}}} {{{e}^{-j\frac{\pi }{\lambda {{d}_{\rm eff}}}{{y}^{2}}}}}dxdy \Bigg|^{2}  
\end{align}
By using integral substitution, we can rewrite
}

 {
\begin{align}
    \notag
    &\tilde{G}_{\rm rect}( {\eta},\varphi)={\left( \frac{\lambda {d}_{\rm eff} \left( 1+{{ {\eta}}^{2}} \right) }{\pi  {\eta} N{{D}^{2}}}\right)}^{2}  \\  \notag
    &\Bigg|\int\limits_{-\sqrt{\frac{\pi }{\lambda {{d}_{\rm eff}}}} \frac{D_{\rm array}}{2\sqrt{1+{{ {\eta}}^{2}}}}+\frac{{{y}_{t}}}{d} \sqrt{\frac{\lambda {{d}_{\rm eff}}}{\pi }}}^{\sqrt{\frac{\pi }{\lambda {{d}_{\rm eff}}}} \frac{D_{\rm array}}{2\sqrt{1+{{ {\eta}}^{2}}}}+\frac{{{y}_{t}}}{d} \sqrt{\frac{\lambda {{d}_{\rm eff}}}{\pi }}}\int\limits_{-\sqrt{\frac{\pi }{\lambda {{d}_{\rm eff}}}} \frac{ {\eta}  D_{\rm array}}{2\sqrt{1+{{ {\eta}}^{2}}}}+\frac{{{x}_{t}}}{d} \sqrt{\frac{\lambda {{d}_{\rm eff}}}{\pi }}}^{\sqrt{\frac{\pi }{\lambda {{d}_{\rm eff}}}} \frac{ {\eta}  D_{\rm array}}{2\sqrt{1+{{ {\eta}}^{2}}}}+\frac{{{x}_{t}}}{d} \sqrt{\frac{\lambda {{d}_{\rm eff}}}{\pi }}}
    \\
    &{{{e}^{-j{{ u }^{2}}}}}  {e}^{-j{{ v }^{2}}}du dv  \Bigg|^{2}, \label{APP_B_1}
\end{align}
where $u=\sqrt{\frac{\pi }{\lambda {{d}_{\rm eff}}}}  x+\frac{{{x}_{t}}}{d} \sqrt{\frac{\lambda {{d}_{\rm eff}}}{\pi }}$, $v=\sqrt{\frac{\pi }{\lambda {{d}_{\rm eff}}}}  y +\frac{{{y}_{t}}}{d} \sqrt{\frac{\lambda {{d}_{\rm eff}}}{\pi }}$, and $d_{\rm eff} = \frac{Fd}{|F-d|} $. By using the Euler formula, i.e., $e^{jx} = {\sf cos}(x)+j {\sf sin}(x)$, we can compute \eqref{APP_B_1} and get 
\begin{align}\label{App_B_2}
    \notag
    &\tilde{G}_{\rm rect}( {\eta},\varphi)={{\left( \frac{1}{4 {\eta}{{p}^{2}}} \right)}^{2}} \times
    \\ \notag
    &\Bigg| \Big( \left( C\left(  {\eta}p+q \right)-C\left( - {\eta}p+q \right) \right)- j\left( S\left(  {\eta}p+q \right)-S\left( - {\eta}p+q \right) \right) \Big) 
    \\ 
    & \Big( \left( C\left( p+ \tilde{q} \right)-C\left( -p+\tilde{q} \right) \right)- j\left( S\left( p+\tilde{q} \right)-S\left( -p+\tilde{q} \right) \right) \Big)      \Bigg|^{2},
\end{align}
where $p=D\sqrt{\frac{N}{4\lambda {{d}_{\rm eff}}(1+{{ {\eta}}^{2}})}}=\frac{1}{2} \sqrt{\frac{{{d}_{FA}}}{{{d}_{\rm eff}}(1+{{ {\eta}}^{2}})}}$, $q=\frac{{{x}_{t}}}{d\pi } \sqrt{2\lambda {{d}_{\rm eff}}} =\frac{\sin(\varphi)}{\pi } \sqrt{2\lambda {{d}_{\rm eff}}}  $, and $\tilde{q}=\frac{{{y}_{t}}}{d\pi } \sqrt{2\lambda {{d}_{\rm eff}}} =\frac{\sin(\theta)}{\pi } \sqrt{2\lambda {{d}_{\rm eff}}}  $.
Since $D = \sqrt{\frac{\lambda d_F}{2}}$ and $C(-v)=-C(v)$ due to the odd function, we can rewrite \eqref{App_B_2} as
}

 {
\begin{align}
    \notag
     &\tilde{G}_{\rm rect}( {\eta},\varphi)={{\left( \frac{1}{4 {\eta}{{p}^{2}}} \right)}^{2}} \times
     \\ \notag
     &\Bigg| \Big( C\left(  {\eta}p+q \right)+C\left(  {\eta}p-q \right)-  j\left( S\left( {\eta}p+q \right)+S\left(  {\eta}p-q \right) \right) \Big) 
     \\ 
     & \Big( C\left( p+\tilde{q} \right)+C\left( p-\tilde{q} \right)-  j\left( S\left( p+\tilde{q} \right)+S\left( p-\tilde{q} \right) \right) \Big) \Bigg|^{2}
\end{align}
Through algebraic simplification, we obtain
\begin{align}
     \notag
     &\tilde{G}_{\rm rect}( {\eta},\varphi)
     =\frac{1}{{{\left( 4 {\eta}{{p}^{2}} \right)}^{2}}} \times
     \\ \notag
     &\left( {{\left( C\left(p+\tilde{q}\right)+C\left( p-\tilde{q} \right) \right)}^{2}}+{{\left( S\left( p+\tilde{q} \right)+S\left( p-\tilde{q} \right) \right)}^{2}} \right)
     \\
     &\Big( {{\left( C\left(  {\eta}p+q \right)+C\left(  {\eta}p-q \right) \right)}^{2}}
     +{{\left( S\left(  {\eta}p+q \right)+S\left(  {\eta}p-q \right) \right)}^{2}} \Big).
\end{align}}
 {
In the case of $y_t = 0$, then $\tilde{q} = 0$, and 
\begin{multline}
\tilde{G}_{\rm rect}( {\eta},\varphi) = \frac{\left( {{\left( C( p) \right)}^{2}}
      +{{\left( S( p )\right)}^{2}} \right)}
      {{{\left( 2 {\eta}{{p}^{2}} \right)}^{2}}} \times \\
\Big( {{\left( C\left(  {\eta}p+q \right)+C\left(  {\eta}p-q \right) \right)}^{2}}
      +{{\left( S\left(  {\eta}p+q \right)+S\left(  {\eta}p-q \right) \right)}^{2}} \Big),  
\end{multline}
which completes the proof.
}

\vspace*{-0.4cm}

\section{Proof of Theorem~\ref{Fresnel_Approx_Circ}}
\label{App_Fresnel_Approx_Circular}

We first substitute the Fresnel approximation in \eqref{eq_III_FresnelApprox} into the normalized array gain in \eqref{eq_II_NormalizedArrayGain} and inject phase-shift $e^{+j \frac{2\pi}{\lambda} \left(\frac{x^2}{2F} + \frac{y^2}{2F}\right)}$. 

Then, we use a polar coordinate system to represent the multiple integrations in \eqref{eq_II_NormalizedArrayGain}, explained in \cite[Section 6]{Riley_BOOK}. We can write the Fresnel approximation of the normalized array gain for the circular array as
\begin{align}
    \notag
    {{\tilde{G}}_{\text{circ}}}&=\frac{1}{{{\left( \pi R^2 \right)}^{2}}}{{\left|  \int\limits_{0}^{2\pi }{\int\limits_{0}^{R} {r {e}^{j\frac{2\pi }{\lambda }\left( F+\frac{{{r}^{2}}}{2F} \right)}  {{e}^{-j\frac{2\pi }{\lambda }\left( z+\frac{{{r}^{2}}}{2z} \right)}}}}drd\varphi  \right|}^{2}} \\ 
    & =\frac{1}{{{\left( \pi R^2 \right)}^{2}}}{{\left| 2\pi \int\limits_{0}^{R}{r {{e}^{-j\frac{\pi }{\lambda } \frac{{{r}^{2}}}{{{z}_{\rm eff}}}}}\ }dr \right|}^{2}},
\end{align}
where $ z_{\rm eff} = \frac{Fz}{|F-z|}$. 
Substituting  $u=-j\frac{\pi {{r}^{2}}}{\lambda {{z}_{\rm eff}}}$, we can rewrite $  {{\tilde{G}}_{\text{circ}}}$ as
\begin{align}
\notag
     &{{\tilde{G}}_{\text{circ}}}={{\left( \frac{1}{\pi {{R}^{2}}} \right)}^{2}}{{\left| j\lambda {{z}_{\rm eff}}\int\limits_{0}^{-j\frac{\pi {{R}^{2}}}{\lambda {{z}_{\rm eff}}}}{{{e}^{u}}du} \right|}^{2}} 
     \\ \notag
     &={{\left( \frac{1}{\pi {{R}^{2}}} \right)}^{2}}{{\left| \lambda {{z}_{\rm eff}}\left( \sin \left( \frac{\pi {{R}^{2}}}{\lambda {{z}_{\rm eff}}} \right)+\left( \cos \left( \frac{\pi {{R}^{2}}}{\lambda {{z}_{\rm eff}}} \right)-1 \right)j \right) \right|}^{2}} \\ 
     &={{\left( \frac{\lambda {{z}_{\rm eff}}}{\pi {{R}^{2}}} \right)}^{2}}\left( 2-2\cos \left( \frac{\pi {{R}^{2}}}{\lambda {{z}_{\rm eff}}} \right) \right) 
     ={{\left( \frac{\sin ( {l})}{ {l}} \right)}^{2}},
\end{align}
where $ {l}= \frac{\pi {{R}^{2}}}{2\lambda {{z}_{\rm eff}}}$. This completes the proof.


\bibliographystyle{IEEEtran}
\bibliography{IEEEabrv,myBib}

\newpage
\begin{IEEEbiography}
 [{\includegraphics[width=1in,height=1.25in,clip,keepaspectratio]{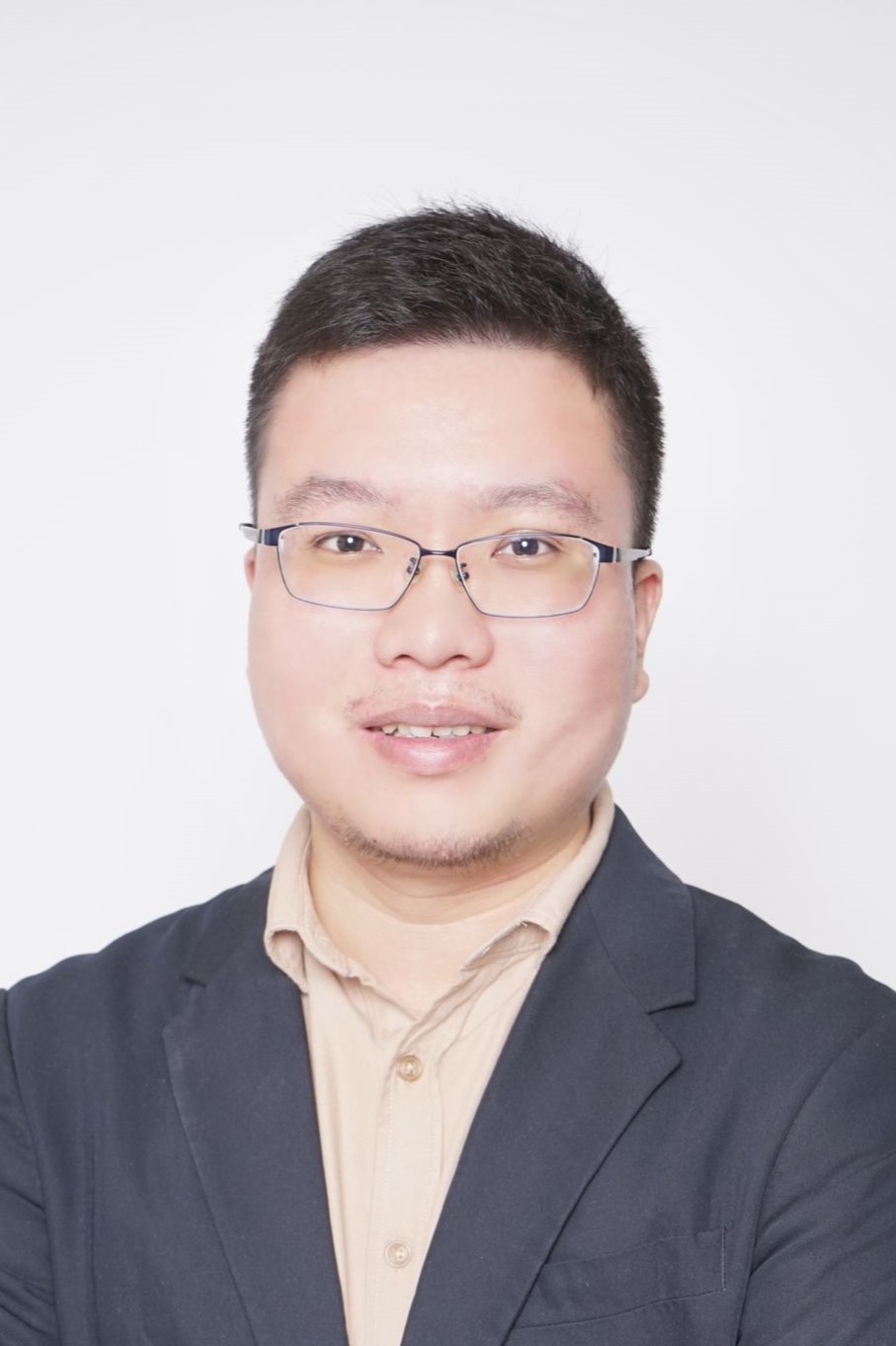}}]{Alva Kosasih}(S’19-M'23)
Alva Kosasih (kosasih@kth.se) is a Postdoctoral researcher at KTH Royal Institute of Technology, Stockholm, Sweden. He holds a Ph.D. in communication engineering from the University of Sydney, Australia (2023). Alva received his B.Eng. and M.Eng. degrees in electrical engineering from Brawijaya University, Indonesia, in 2013 and 2017, respectively, and an M.S. degree in communication engineering from the National Sun Yat-sen University, Taiwan, in 2017. His research interests include signal processing for extra large-scale MIMO, near-field communications, MIMO detection, and machine learning for the physical layer.
\end{IEEEbiography}
\vspace{-8cm}
\begin{IEEEbiography}
[{\includegraphics[width=1in,height=1.25in,clip,keepaspectratio]{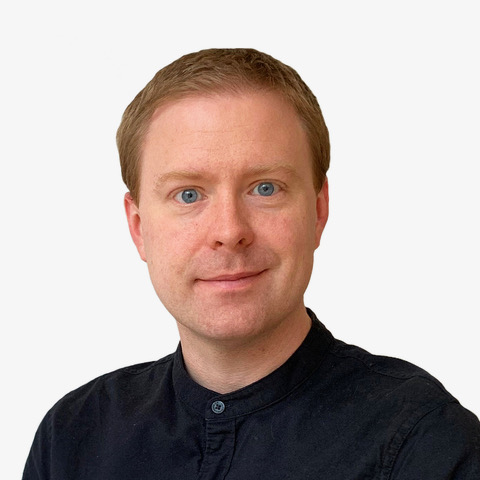}}]{Emil Bj\"ornson}(S'07-M'12-SM'17-F'22)
Emil Bj\"ornson is a Full Professor of Wireless Communication at the KTH Royal Institute of Technology, Sweden. He received an M.S. degree in engineering mathematics from Lund University, Sweden, in 2007, and a Ph.D. degree in telecommunications from KTH in 2011. From 2012 to 2014, he was a post-doc at the Alcatel-Lucent Chair on Flexible Radio, SUPELEC, France. From 2014 to 2021, he held different professor positions at Link\"oping University, Sweden. He was a Visiting Full Professor at KTH in 2020-2021, before obtaining a tenured position in 2022.

He has authored the textbooks \emph{Optimal Resource Allocation in Coordinated Multi-Cell Systems} (2013), \emph{Massive MIMO Networks: Spectral, Energy, and Hardware Efficiency} (2017), \emph{Foundations of User-Centric Cell-Free Massive MIMO} (2021), and \emph{Introduction to Multiple Antenna Communications and Reconfigurable Surfaces} (2024). He is dedicated to reproducible research and has made a large amount of simulation code publicly available. He performs research on MIMO communications, radio resource allocation, machine learning for communications, and energy efficiency. He has been an associate editor for multiple IEEE transactions and magazines.

He has performed MIMO research for 18 years, his papers have received more than 28000 citations, and he has filed more than 20 patent applications. He is a host of the podcast Wireless Future and has a popular YouTube channel with the same name. He is an IEEE Fellow, a Wallenberg Academy Fellow, a Digital Futures Fellow, and an SSF Future Research Leader. He has received the 2014 Outstanding Young Researcher Award from IEEE ComSoc EMEA, the 2015 Ingvar Carlsson Award, the 2016 Best Ph.D. Award from EURASIP, the 2018 and 2022 IEEE Marconi Prize Paper Awards in Wireless Communications, the 2023 IEEE ComSoc Outstanding Paper Award, the 2019 EURASIP Early Career Award, the 2019 IEEE Communications Society Fred W. Ellersick Prize, the 2019 IEEE Signal Processing Magazine Best Column Award, the 2020 Pierre-Simon Laplace Early Career Technical Achievement Award, the 2020 CTTC Early Achievement Award, and the 2021 IEEE ComSoc RCC Early Achievement Award. He also co-authored papers that received Best Paper Awards at the conferences, including WCSP 2009, the IEEE CAMSAP 2011, the IEEE SAM 2014, the IEEE WCNC 2014, the IEEE ICC 2015, and WCSP 2017.
\end{IEEEbiography}

\end{document}